\documentclass[10pt,journal,compsoc]{IEEEtran}
\usepackage{tikz}
\usepackage{subfigure}
\usepackage{mathtools}
\usepackage{enumitem}
\usepackage{hyperref}
\usepackage{amsmath,amssymb,amsfonts}
\usepackage{graphicx}
\usepackage{textcomp, hyperref}
\usepackage{xcolor}
\usepackage{enumitem}
\usepackage{amsthm}
\usepackage{tikz}
\usepackage{blindtext}
\usepackage{amsmath}
\usepackage{subfigure}
\usepackage{makecell}
\usepackage{epstopdf,comment}
\usepackage{wrapfig}
\newtheorem{lemma}{Lemma}

\newtheorem{theorem}{Theorem}

\newtheorem{corollary}{Corollary}
\usepackage{mathtools}

\usetikzlibrary{shapes.geometric}
\newtheorem{remark}{Remark}
 
 \newcommand{\ly}[1]{\textcolor{black}{#1}}

\ifCLASSOPTIONcompsoc
  \usepackage[nocompress]{cite}
\else
  \usepackage{cite}
\fi
\ifCLASSINFOpdf
\else
\fi
\begin{document}
\title{On the Effect of Bounded Rationality\\ in Mixed Duopoly Markets}

\author{Lihui Yi and Ermin Wei% <-this % stops a space
\IEEEcompsocitemizethanks{\IEEEcompsocthanksitem The authors are with Northwestern University, Evanston,
IL, 60208. L. Yi is with the Department
of Electrical and Computer Engineering (e-mail: lihuiyi2027@u.northwestern.edu).
 E. Wei is with the Department
of Electrical and Computer Engineering and Department of Industrial Engineering and Management Sciences (e-mail: ermin.wei@northwestern.edu).\protect\\
}% <-this % stops an unwanted space
\thanks{This work was supported by the National Science Foundation (NSF) under Grants CNS-2030251 and ECCS-2216970 and the Center for Engineering Sustainability and Resilience at Northwestern Engineering.}}


\markboth{Journal of \LaTeX\ Class Files,~Vol.~14, No.~8, August~2015}%
{Shell \MakeLowercase{\textit{et al.}}: Bare Demo of IEEEtran.cls for Computer Society Journals}



\IEEEtitleabstractindextext{%
\begin{abstract}
Nash equilibrium is a common solution concept that captures strategic interaction in market analysis. \ly{However, it assumes that all market participants are fully rational and can compute mutually consistent best responses. Level-\textit{k} reasoning relaxes this assumption by modeling bounded rationality through finite reasoning depth.} In this paper, {we consider a variation of Cournot competition that includes two suppliers,} both following level-\textit{k} reasoning. One is a self-interested firm and the other serves as a benevolent social planner. \ly{First, we show that social welfare is maximized when the planner is one level more rational than its opponent. Welfare is weakly decreasing as the planner's level moves away from this optimum.} \ly{We then show that welfare under bounded rationality can differ substantially from welfare at Nash equilibrium.} \ly{We characterize optimal, expectation-maximizing, and robust maximin strategies for the planner under different amounts of information about its opponent's rationality level. Finally, we show how utility design can improve social welfare by allowing the planner to cooperate with or fight the self-interested firm.} {Numerical} experiments further demonstrate and validate our findings.
\end{abstract}

}

\maketitle

\IEEEdisplaynontitleabstractindextext

% \fi
\IEEEpeerreviewmaketitle

\IEEEraisesectionheading{\section{Introduction}\label{sec:introduction}}


\IEEEPARstart{U}nderstanding the interactions between market participants is crucial in many engineered and economic systems -- e.g., electricity markets, transportation, and postal services -- where game theory has been widely adopted to {model} strategic behaviors of competing firms, predict equilibrium outcomes, and evaluate market efficiency \cite{Abapour, Peeters2017}. {In particular, game-theoretic analysis has played an essential role in structuring {wholesale electricity markets in the United States} \cite{einhorn_2012_from}. Classical studies demonstrate how strategic generation firms' behavior under network and operational constraints can significantly affect prices and welfare outcomes, making game-theoretic analysis a central tool in electricity market design and regulation \cite{Borenstein2002, Wilson2002, KAHN200170}. }

Game-theoretic models, and in particular equilibrium analysis, have long served as a primary tool for analyzing {such} systems \cite{Fudenberg2016}. \ly{In these models, Nash equilibrium is a common solution concept. It is a strategy profile} \ly{in which no participant has an incentive to deviate unilaterally}. However, it implicitly assumes that all participants are perfectly rational and capable of computing mutually consistent best responses \cite{Brandl2024, Daskalakis2009}.

In practice, this assumption is often unrealistic \cite{Flood, Simon, Dhami}. Market participants usually have limited cognitive abilities and computational resources \cite{Milionis2023}. {For example, in electricity markets, computing equilibrium strategies requires forecasting demand, network congestion, and competitors' actions under complex operational constraints, making fully rational equilibrium behavior difficult to achieve in practice.} As a result, their behaviors may systematically deviate from the Nash equilibrium. These observations have motivated a substantial body of work on bounded rationality, which seeks to capture more realistic decision-making processes. Among the various approaches proposed, level-$k$ reasoning has emerged as a particularly intuitive and tractable framework. In level-$k$ reasoning, people are differentiated by their depth of strategic thinking:

\textit{{A} level-0 agent plays a random strategy; {a} level-1 agent best responds {to} a level-0 agent; {a} level-2 agent best responds to a level-1 agent; ...; {a} level-$k$ agent best responds to a level-$(k-1)$ agent.}

{Naturally, the reasoning depth $k$ can be regarded as a measure of a player's level of rationality. \ly{A higher value of $k$ corresponds to deeper iterative reasoning and represents greater cognitive and computational capacity to analyze strategic interactions.} }

While bounded rationality has been studied extensively in standard oligopoly models \cite{Awad2024, li_2024_dynamics, BENNDORF2021102773}, \ly{little attention has been paid} to settings where a profit-maximizing firm competes against a welfare-oriented benevolent social planner. Such structures arise naturally in markets where public entities participate alongside private firms, especially in markets where public goods or services are involved. Such markets, sometimes called {\it mixed duopolies } \cite{DeFraja1089}, are important in a wide range of industries \cite{Gil-Moltó2011}. They introduce an intrinsic asymmetry in objectives, where the self-interested firm maximizes profit and the social planner aims to ensure that the market operates efficiently and equitably, taking into account the interests of all stakeholders. \ly{Existing analyses of mixed duopolies largely rely on Nash equilibrium}, thereby presuming full rationality on both sides. This raises a fundamental question: how should welfare outcomes and planner behavior be understood \ly{when both market participants have bounded rationality, different objectives, and different reasoning depths}? {Moreover, many engineered markets--including electricity systems--are fundamentally constrained by network infrastructure such as transmission line capacities. Understanding how such network constraints affect \ly{behavior under bounded rationality becomes essential} for predicting welfare outcomes and designing strategies. }

{Our preliminary work \cite{prelimCDC} models an electricity market and takes a first step toward addressing this problem, by using level-$k$ reasoning to study the effect of bounded rationality in a network-constrained mixed duopoly market \ly{under a specific consumer utility function.} In this paper, we continue using electricity markets as a motivating example and substantially extend the previous work in several directions. First, we generalize the model to accommodate any quadratic consumer utility function, \ly{which broadens the scope of the analysis}.} {\ly{Second, we study utility design from the planner's perspective. We ask whether an alternative, interpretable utility function can improve realized social welfare.} To investigate this question, we allow the planner's utility to be constructed as a linear combination of social welfare and the self-interested firm's profit, where the latter is weighted by a design parameter $\gamma$ that captures the degree of cooperation between the two producers. When $\gamma = 0$, the planner maximizes social welfare directly. When $\gamma$ is positive or negative, the planner cooperates with or fights the self-interested firm, respectively. \ly{We show that an appropriate cooperation level $\gamma$ can improve realized social welfare compared with directly maximizing welfare under level-$k$ beliefs.} \ly{We choose this utility structure for its analytical tractability and clear interpretation, following \cite{Cai}.} While alternative utility function constructions are possible, we restrict our attention to this formulation to gain clear insights.}

\ly{We address four questions under level-$k$ reasoning:}

\vspace{5pt}
\noindent \textbf{Question 1.} \ly{Does a higher rationality level of the social planner always improve social welfare?}

\vspace{5pt}
\noindent \textbf{Question 2.} How does {the} level-$k$ reasoning outcome deviate from the Nash equilibrium outcome? 

\vspace{5pt}
\noindent \textbf{Question 3.} \ly{What strategies can the planner use with different amounts of information about its opponent's rationality level?}

\vspace{5pt}
\noindent \textbf{Question 4.} \ly{Can social welfare improve if the social planner cooperates with or fights the self-interested firm?}

\vspace{5pt}

\ly{Our main contributions answer these four questions:}
\begin{itemize}
    \item First, {we adopt level-$k$ reasoning in Cournot competition to model \ly{suppliers' behavior under bounded rationality}.} \ly{We derive closed-form expressions for their level-$k$ behavior and show that welfare is maximized when the planner is one level more rational than its opponent. Welfare is weakly decreasing as the planner's level moves away from this optimum in either direction.}
    \item Second, we find that the social welfare outcome of {the} level-$k$ reasoning model \ly{may be better than, equal to, or worse than} that of the Nash equilibrium, depending on the network {constraints}. Our numerical results show that bounded rationality plays a crucial role in system outcomes.
    \item Third, we propose strategies for the benevolent social planner \ly{to respond to the self-interested firm with complete, incomplete, or no information about its rationality level}. \ly{We compare their performance numerically.}
    \item \ly{Finally, we design the planner's utility function and show that cooperation with or opposition to the self-interested firm can improve realized social welfare. Under level-$k$ beliefs, a utility function that differs from social welfare can therefore yield higher welfare.}
\end{itemize}

The remainder of the paper is organized as follows. Section \ref{sec: related work} reviews the related literature and compares our work with the existing studies. Section \ref{sec: model} introduces the {network-constrained} mixed Cournot duopoly model and motivates the model through an electricity market example. Section \ref{sec: level-k reasoning} formalizes level-$k$ reasoning in our mixed Cournot duopoly model and derives the best-response behaviors for both market participants. Section \ref{sec: level k welfare performance} analyzes the welfare performance of level-$k$ outcomes and characterizes how welfare depends on the relative bounded rationality levels of the two players. Section \ref{sec: level k vs NE} \ly{compares outcomes under bounded rationality with the fully rational Nash equilibrium}. Section \ref{sec: planner strategies} \ly{studies the planner's strategies} under different information regimes, including complete, incomplete, and no information about {the rationality level of} the self-interested firm. Section \ref{sec: utility design} \ly{investigates how the planner can improve social welfare through utility design.} Section \ref{sec: numerical} \ly{presents numerical results that illustrate the analysis}. Finally, Section \ref{sec: conclusion} concludes the paper and discusses directions for future research.

{Supplementary proofs are provided in Appendices \ref{appendix: closed-form level k behaviors}--\ref{appendix: utility design PoR}.}

\section{Related Work}\label{sec: related work}

This paper \ly{lies at the intersection of} mixed oligopoly markets and behavioral game theory. We adopt the Cournot competition model, which provides a widely used foundation for analyzing strategic quantity competition in oligopoly markets, and study the effect of bounded rationality, which suggests that Nash equilibrium may not be reliable in predicting strategic behaviors. In this section, we review several lines of literature--Cournot competition, mixed oligopoly markets, bounded rationality models, and level-$k$ reasoning. \ly{Finally, we explain how our work relates to and differs from these studies.}

\subsection{Cournot competition}
While various game-theoretic frameworks {have been} applied across different market settings (e.g., Bertrand competition \cite{Janssen2002, Yuanzhang}, Stackelberg game \cite{Maharjan2013, Lang, VijayGupta}, etc.), Cournot competition \cite{ALLAZ19931} remains one of the most {fundamental} approaches. It is a classical economic model that analyzes strategic behavior among firms within an oligopoly, \ly{in which firms choose production quantities and total output determines the market price.}

Cournot competition is widely used in market modeling not only because it is analytically tractable, but also because it captures a fundamental feature of oligopolistic interactions: firms recognize that their production output affects market price and thus choose quantities strategically rather than behaving as price takers \cite{Kreps1983}. 

It is also valuable because it \ly{provides insight into} various real-world markets. For example,
in telecommunications, Cournot oligopoly models have been used to study how spectrum scarcity affects competitive outcomes in mobile markets \cite{GRUBER200159}. 
\ly{In industrial organization,} Cournot-based models are used as a simulation tool to evaluate how firms' acquisitions affect prices and welfare in gasoline markets \cite{greenfield_2015_simulating}.
In electricity markets, Cournot models can be used to quantify the potential for market power and welfare loss under different market structures \cite{Borenstein1999, NEUHOFF2005495}. 
Hence, Cournot competition serves as a standard modeling tool for analyzing oligopolistic behaviors in real-world markets.

\subsection{Mixed oligopoly markets}

\ly{In a standard oligopoly, all firms compete to maximize their own profits. A mixed oligopoly also includes a public entity or benevolent social planner that aims to maximize social welfare.}

\ly{Existing studies use Cournot competition to analyze the effects of social planners in several settings.} Specifically,
Kalashnikov \textit{et al.} \cite{Kalashnikov01072010} \ly{investigate a Cournot mixed duopoly model and establish existence and uniqueness of equilibrium.}
Haraguchi and Matsumura \cite{haraguchi_2015_cournotbertrand} discuss the Cournot and Bertrand competitions in a mixed oligopoly and \ly{compare the resulting profits}.
Bose \textit{et al.} \cite{Bose2014} \ly{consider a Cournot model with multiple firms that compete locally, together with a simultaneous-move game between the social planner and the firms.} They show that the existence and the nature of market equilibria are significantly impacted by the social planner's objective. Cai \textit{et al.} \cite{Cai} extend the work and further demonstrate a counterintuitive observation that, if the social planner intends to optimize social welfare, it should not use social welfare as the objective function.
{More recent work continues to revisit mixed oligopoly markets under additional strategic considerations, such as collusion \cite{escrihuelavillar_2024_on}. }

However, despite the rich literature on equilibrium-based analysis, existing works overwhelmingly adopt Nash equilibrium, implicitly assuming that both the profit-maximizing firm and the public entity are fully rational and have unlimited computational resources to calculate the equilibrium.

\subsection{Bounded rationality models}

A large body of literature has argued that full rationality and Nash equilibrium can be unrealistic in modeling strategic behaviors, both theoretically and empirically. From a theoretical perspective, Simon \cite{Simon1955} points out that real decision-makers face cognitive and informational limits, and therefore rely on simplified rules rather than finding best responses. In addition, even if agents have unlimited cognitive ability to act fully rationally, computing a Nash equilibrium is PPAD-complete in general games \cite{Daskalakis2009}, including finite games with only a small number of players. This indicates that no efficient algorithm is likely to exist for the equilibrium computation in the worst case \cite{Chen2009}. Moreover, classic equilibrium-finding algorithms, \ly{such as the Lemke-Howson method \cite{Lemke1964}, have been shown} to exhibit exponential running time on some bimatrix games, even in the best case \cite{Savani2006}. \ly{These results show that the computational difficulty of finding a Nash equilibrium is not merely theoretical but can arise in practice.}

Complementing the theoretical perspective, there is substantial experimental evidence demonstrating systematic deviations from Nash equilibrium as well. \ly{Güth \textit{et al.} \cite{GUTH1982367} conduct ultimatum bargaining experiments and observe offers and rejections that differ significantly from subgame-perfect equilibrium outcomes.} \ly{Nagel \cite{Nagel} studies beauty-contest games and finds that choices cluster around intermediate reasoning levels rather than the Nash prediction.} \ly{More broadly, Camerer \cite{Camerer2003} reviews hundreds of experiments across a wide range of games that show deviations from equilibrium behavior.} Together, these theoretical and empirical findings motivate the development of alternative models that relax the assumption of fully rational equilibrium plays while retaining the strategic interactions.

\ly{Bounded rationality can be modeled in many ways. An early approach replaces exact optimization with simple decision rules: players seek ``good enough'' decisions under constraints rather than exact best responses} \cite{Simon1955}. A closely related line of work introduces stochastic choice: players choose strategies based on relative expected utility and assume other players do so as well. This leads to solution concepts such as \ly{quantal response equilibrium \cite{MCKELVEY19956, plott_2008_handbook, Lang2}, which allow systematic deviations from exact best responses.}

\ly{Another major family of models of bounded rationality considers learning and adaptation over time.} In these models, \ly{players do not compute an equilibrium in advance}. Instead, they adjust their actions gradually based on past experience and observed payoffs. In particular,
\ly{Erev and Roth \cite{Erev1998} use reinforcement learning, which assigns higher choice probabilities to strategies with higher past payoffs. Their model more closely matches observed behavior in repeated games than Nash equilibrium predictions do.}
Fudenberg and Levine \cite{FUDENBERG1998631} propose belief-based learning models that allow agents to form and update beliefs about opponents' actions, and best respond to those beliefs. Depending on the game structure, the long-term dynamics may or may not converge to equilibrium asymptotically.

A complementary perspective models bounded rationality through reasoning depth, focusing on how players form beliefs about others' reasoning, rather than how they learn from experience. Specifically, agents do not assume mutual consistency of beliefs, but instead \ly{reason through only a finite number of steps. This approach} is formalized as level-$k$ reasoning and its generalization, cognitive hierarchy models, initially conceptualized by Nagel \cite{Nagel} in 1995. In these models, players best respond to beliefs that their opponents are less sophisticated and have shallower reasoning depths \cite{STAHL1995218, Camerer2004}. Unlike stochastic choice or learning-based models, level-$k$ and cognitive hierarchy explicitly capture the heterogeneity in strategic sophistication among players, and have strong experimental support across a wide range of games. \ly{These models have therefore become popular alternatives to Nash equilibrium. Next, we review the level-$k$ and cognitive hierarchy literature in more detail.}

\subsection{Level-\textit{k} reasoning}
The conceptual foundation of level-$k$ reasoning is derived from experimental evidence in one-shot games. In beauty contest experiments, Nagel \cite{Nagel} \ly{finds that subjects' choices consistently cluster} around outcomes corresponding to one or two iterations of reasoning, rather than the Nash equilibrium obtained by infinite iterations. Based on this insight, level-$k$ reasoning is proposed, which classifies players into discrete reasoning depths, such as level-0 (non-strategic), level-1 (best responding to level-0), and higher levels defined recursively. \ly{Camerer \textit{et al.} \cite{Camerer2004} further generalize} this framework by allowing players to believe that opponents are distributed across lower reasoning levels, instead of fully concentrated at one single level.

{Level-$k$ reasoning is {well supported by many behavioral} experiments \cite{Stahl, CH_Model, Lk_model, lk_experiment, Bertrand_electricity}. For example, 
Camerer \textit{et al.} \cite{CH_Model} \ly{observe that level-1 and level-2 behavior is common};
Arad and Rubinstein \cite{lk_experiment} find that 84\% of the subjects chose the strategies corresponding to level-0 through level-3 behaviors. 
Horta\c{c}su \textit{et al.}\cite{Bertrand_electricity} \ly{apply a cognitive hierarchy model with Bertrand competition to Texas electricity markets and find that it explains observed behavior better than Nash equilibrium does.} In level-$k$ reasoning, the rationality of an agent is naturally modeled by the level $k$, as more cognitive ability and computational power are required for a larger $k$. {As $k$ approaches infinity, level-$k$ behaviors may converge to Nash equilibrium (e.g., $p$-beauty contest game and our game in this paper), where agents are fully rational, and no one can do better by unilaterally changing {their} strategy.} }

\begin{table*}[t]
\centering
\caption{{Comparison of Cournot-based models, assumptions, and scope}}
\label{tab:cournot_comparison}
\begin{tabular}{lcccccccc}
\hline
\textbf{Work} 
& \makecell{\textbf{Market}\\\textbf{Structure}}
& \makecell{\textbf{Firms'}\\\textbf{Objectives}}
& \makecell{\textbf{Solution}\\\textbf{Concept}}
& \makecell{\textbf{Bounded}\\\textbf{Rationality}}
& \makecell{\textbf{Welfare/Strategy}\\\textbf{Analysis}}
& \makecell{\textbf{Utility}\\\textbf{Design}} 
& \makecell{\textbf{Network}\\\textbf{Constraints}} \\

\hline
Standard Cournot Duopoly 
& Duopoly 
& Profit--Profit
& Nash 
& $\times$ 
& $-$ 
& $-$ 
& $\checkmark$ \\

Mixed Cournot Duopoly
& Mixed Duopoly
& Profit--Welfare
& Nash 
& $\times$ 
& $-$ 
& $-$ 
& $\checkmark$ \\

Runco \cite{Runco}
& Duopoly 
& Profit--Profit
& Level-$k$ 
& $\checkmark$ 
& $\times$ 
& $\times$
& $\times$ \\

% Preliminary Work \cite{prelimCDC}
% & Mixed Duopoly
% & Profit--Welfare
% & Level-$k$ 
% & $\checkmark$ 
% & Partial 
% & $\times$ \\

\textbf{This Work} 
& Mixed Duopoly 
& Profit--Welfare
& Level-$k$ 
& $\checkmark$ 
& $\checkmark$ 
& $\checkmark$ 
& $\checkmark$ \\

\hline
\end{tabular}
\end{table*}

\subsection{Research gaps and our contribution}

Our work is closely related but distinct from the literature reviewed above. At a high level, existing studies establish: 1) {Cournot competition as a powerful model} for capturing strategic interactions in real-world markets; 2) mixed duopoly as a framework for studying \ly{competition between public and private firms} and welfare outcomes; 3) bounded rationality and level-$k$ reasoning \ly{as well-supported alternatives to Nash equilibrium analysis}. \ly{We combine these elements in one analytical framework to study welfare implications and derive strategies that maximize welfare.}

While bounded rationality has been explored in standard Cournot duopolies \cite{Runco}, its impact on mixed duopolies remains an open question. This gap is significant because mixed markets feature a fundamental difference in payoff functions: a self-interested firm competes against a benevolent social planner. {As a result, \ly{questions about social welfare, the planner's strategies, and utility design} remain outside the scope of \cite{Runco}. }

Our preliminary work \cite{prelimCDC} takes a first step towards addressing these issues by studying mixed duopoly markets and examining bounded rationality effects \ly{under a specific consumer utility function}. However, the earlier analysis is limited in scope and does not consider the role of utility design.

The present paper extends both lines of work {in several ways.} First, we study level-$k$ reasoning in a mixed Cournot duopoly, where firms' objectives are asymmetric and therefore the strategic interactions are more complicated. Second, we consider \ly{a general quadratic consumer utility function}, and \ly{analyze welfare outcomes and the planner's strategies} \ly{with different information about the self-interested firm's rationality level.} Finally, we study the question of utility design and show that \ly{an appropriately designed planner objective} can improve the realized social welfare.

Table \ref{tab:cournot_comparison} summarizes how the present work differs from the closest Cournot-based models in terms of market structure, {assumptions, and scope.}

\section{Model}\label{sec: model}

{In this section, we present our model--Cournot competition in mixed duopoly markets. \ly{First, we introduce the demand and supply sides of the market. We then illustrate the model with an electricity market example.} }

\subsection{Demand side}
{\ly{Consumers buy a single divisible commodity from the suppliers. We represent aggregate demand by a consumer with a quadratic utility function $u(\cdot)$, as commonly assumed in Cournot competition} \cite{Cai,Bimpikis2019},}
$${u(q_D) = -\frac{a}{2}q_D^2 + bq_D + m},$$
for some $a,b > 0$ and $m \ge 0$. {Here, we restrict the demand $q_D \in [0,b/a]$ to ensure that $u(q_D)$ is non-decreasing.}

The consumer aims to maximize the net benefit, i.e., utility {minus} payment to the suppliers, 
$$\pi_D(q_D) = u(q_D) - pq_D,$$
where $p$ is the market price. Since $\pi_D(q_D)$ is concave, the optimal demand satisfies the first-order condition, which leads to
\begin{equation}\label{eq: price}
    {p(q_D)=b-aq_D,}
\end{equation}
i.e., the market exhibits a linear inverse demand function as seen in the {classical} Cournot competition.

\subsection{Supply side}
{Suppose there are two firms serving as suppliers in the market. \ly{Firm $S$ is self-interested and maximizes its own profit. Firm $B$, also called the social planner, is benevolent and maximizes social welfare.} } We assume that {the} two firms have identical linear production {costs} \cite{linear_cost} as $c_i(q_i) = c q_i$, where {$c \in [0,b)$} is the cost per unit and {$q_i \in [0,b/a]$} is the quantity supplied by firm $i$ with $i \in \{S,B\}$.

The payoff (or utility) of the self-interested firm $S$ is defined as its net profit, given by
$$\pi_S(q_S,q_B) = p(q_D)q_S - c q_S,$$
where $p(q_D)$ is given by Eq. \eqref{eq: price}.
The payoff (or utility) of the benevolent social planner $B$ is defined as the system social welfare $W(q_S,q_B)$, i.e., 
$$\pi_B(q_S,q_B) = W(q_S,q_B),$$
where $W(q_S,q_B) = u(q_D) - cq_S - cq_B$. \ly{Both firms aim to maximize their respective payoffs, and the market clears: $q_D=q_S+q_B$.}

{Furthermore, we assume that the social planner has limited market power, characterized by a production capacity $f \le (b-c)/a$. \ly{With sufficient capacity, the social planner can supply the welfare-optimal total output {$(b-c)/a$} by itself, and the self-interested firm produces zero at equilibrium. The capacity constraint allows us to study competition when the planner cannot serve the market alone.} This production capacity also reflects operational constraints commonly faced by \ly{welfare-oriented entities} due to financial limitations \cite{HUANG2022108374}, regulatory approvals \cite{best2018social}, or physical infrastructure constraints \cite{BARRAKET2019188}. }

\subsection{An example in electricity markets}

\ly{Our mixed duopoly model can represent an electricity market, where Cournot competition is widely used} \cite{Gountis, Cunningham, Baosen, Cai}. In electricity markets, the consumer (or demand) $D$ buys electricity from the power suppliers, {with a utility commonly assumed} to be quadratic in the power system literature \cite{quad_utility, quad_utility2, quad_utility3}. The two firms $S$ and $B$ serve as the power suppliers, with the self-interested firm interpreted as a private generator and the social planner interpreted as {a publicly owned entity} \cite{public_sector}. The linear production cost is justified by Bose and Bitar \cite{linear_cost}.

In electricity markets, one can interpret the production capacity $f$ as the transmission line limit, as depicted in the 3-bus network example in Figure \ref{fig: topology 1}. As seen in the figure, firm $B$ and firm $S$ are two generators at buses 1 and 2, respectively, and the consumer is a load at bus 3. We assume the transmission line {between buses 1 and 2} has a maximum flow capacity $f$ and the line {between buses 2 and 3} has a flow capacity $2b/a$. \ly{Under the assumptions in the previous two subsections, each supplier produces at most $b/a$ units. Thus, flow on line 2-3 never exceeds $2b/a$.} Essentially, the flow capacity of line 2-3 is always sufficient. Only line 1-2 is possibly congested, as also seen in practice (e.g., Path 15 in California \cite{Path15}).

\begin{figure}[ht]
\centering
\includegraphics[width=0.8\linewidth]{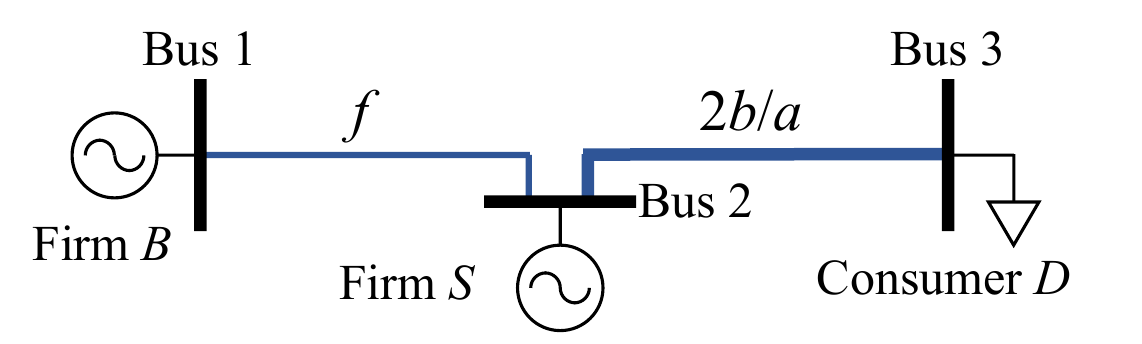}
\caption{A 3-bus network example.}
\label{fig: topology 1}
\end{figure}

\ly{Under the direct current (DC) power flow model, commonly used in electricity market studies \cite{Yuanzhang}, this simple network has the following properties:}

\begin{enumerate}
    \item The system is lossless;
    \item \ly{The power flows from bus 1 to bus 2 and from bus 2 to bus 3 are $q_B$ and $q_B+q_S$, respectively;}
    \item \ly{The power flow equation constraint is $|q_B| \le f$.}
\end{enumerate}
\ly{These properties ensure that all electricity produced by the suppliers reaches the consumer.} {Note that the results also apply to other topologies under DC assumptions, such as the 2-bus and 3-bus systems depicted in Figure \ref{fig: other topology}.
\begin{figure}[ht]
    \centering
    \subfigure[A 2-bus system]{
        \begin{minipage}[t]{0.5\linewidth}
        \centering
        \includegraphics[width=\linewidth]{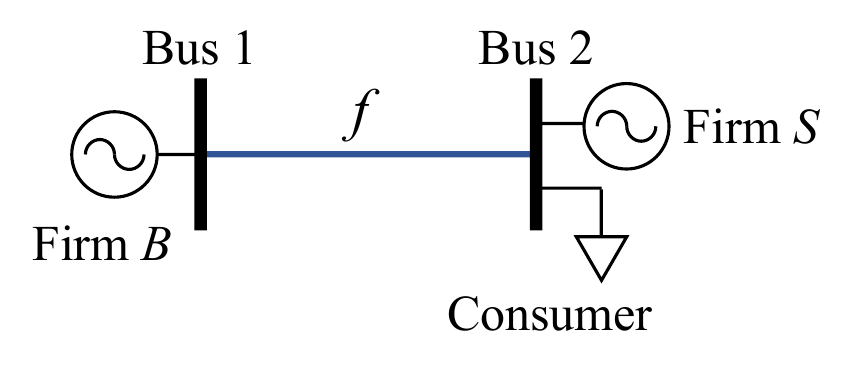}
    \end{minipage}
    }%
    \subfigure[A 3-bus system]{
    \begin{minipage}[t]{0.5\linewidth}
        \centering
        \includegraphics[width=\linewidth]{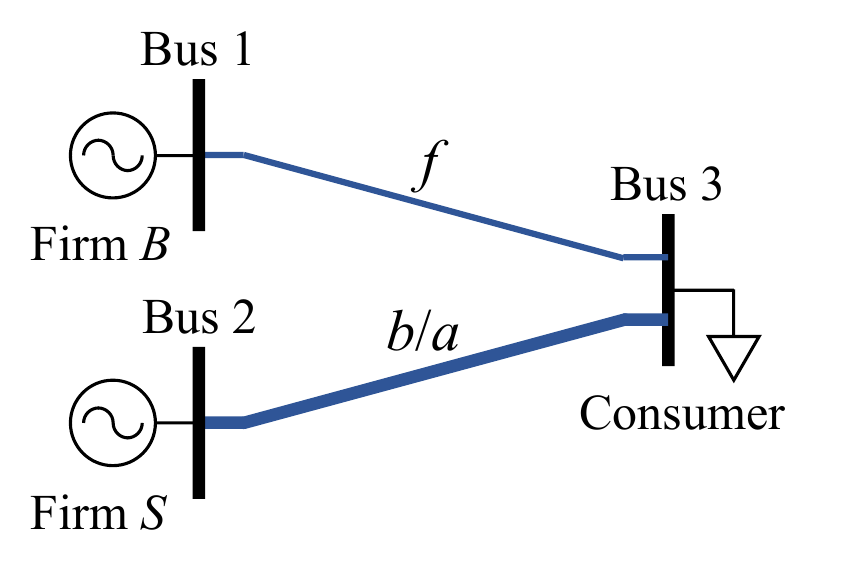}
    \end{minipage}
    }
    \caption{Other network topologies.}
    \label{fig: other topology}
\end{figure}}

\section{Level-\textit{k} Reasoning}\label{sec: level-k reasoning}
\ly{We now define level-$k$ reasoning for the mixed duopoly.} As in the classical Cournot competition, {we define the \textit{best response} of firm $S$ to firm $B$ and of firm $B$ to firm $S$ as $\mathcal{B}_S(q_B)$ and $\mathcal{B}_B(q_S)$, respectively, such that $\mathcal{B}_S(q_B) = \arg\max_{q_S \ge 0} \pi_{S}(q_S, q_B)$ and $\mathcal{B}_B(q_S) = \arg\max_{0 \le q_B \le f} \pi_{B}(q_S, q_B)$ due to {the} {capacity} constraint.}
\ly{The next lemma gives closed-form expressions for both suppliers' best responses.}

\begin{lemma}\label{lemma: best response functions}
    \ly{The best responses of the two suppliers are:}
    \begin{equation*}
        \begin{aligned}
            &\mathcal{B}_S(q_B) = \max \Big\{ (b-c-aq_B)/2a, 0 \Big\},\\
            &\mathcal{B}_B(q_S) = \max \Big\{ \min\big\{ (b-c-aq_S)/a, f \big\}, 0 \Big\}.
        \end{aligned}
    \end{equation*}
\end{lemma}
\begin{proof}
    \ly{We rewrite the payoff functions of the two suppliers as follows:}
    \begin{equation*}
        \begin{aligned}
            &\pi_{S} = -aq_S^2 + (b-c-aq_B)q_S,\\
            &\pi_{B} = -\frac{a}{2}q_B^2 \!+\! (b\!-\!c\!-\!aq_S)q_B \!+\! q_S(b\!-\!c\!-\!\frac{aq_S}{2}) \!+\! m.
        \end{aligned}
    \end{equation*}
    \ly{Since $\pi_S$ is concave in $q_S$ and $\pi_B$ is concave in $q_B$, the first-order conditions and output constraints give the stated best responses.}
\end{proof}

As a demonstration of the level-$k$ reasoning model, we present level-0, level-1 and level-2 behaviors as examples in our framework and then give the general formula {for} level-$k$ behaviors. Denote the quantity produced by the level-$k$ firm $i$ by $q_{i}^{(k)}$, with $i \in \{S,B\}$ and $k \in \mathbb{N}$ \footnote{We consider $\mathbb{N}$ to include 0.}.

\textit{i) Level-0 behavior:}
{Level-0 players choose feasible quantities at random \cite{L0}. We represent this behavior by the mean quantity for each firm.} For the self-interested firm $S$, the minimal production quantity is 0 while the maximum is {$(b-c)/a$}, since any extra produced unit will result in a price lower than the marginal cost $c$. {Here, we simply let {$q_S^{(0)} = (b-c)/2a$}, which is the mean of all feasible solutions \cite{L0}\footnote{{Alternatively, level-0 quantities can follow uniform distributions over their feasible intervals, with level-1 agents maximizing expected payoffs. In this model, the resulting level-1 best responses equal the best responses to the mean quantities.}}.} For the benevolent social planner $B$, the minimal and maximal production quantities are 0 and $f$, respectively. Thus, {we let $q_B^{(0)} = f/2$.}

\textit{ii) Level-1 behavior:}
The level-1 players assume the other players are level-0 and best respond accordingly. Specifically, the level-1 firm $S$ will best respond to the level-0 firm $B$. {Substituting the level-0 quantity into firm $S$'s best response gives}
$${q_S^{(1)} = \mathcal{B}_S(q_B^{(0)}) = (b-c)/2a - f/4}.$$
Similarly, the level-1 firm $B$ will best respond to the level-0 firm $S$ and hence
$${q_B^{(1)} = \mathcal{B}_B(q_S^{(0)}) = \min\left\{ (b-c)/2a, f \right\}}.$$

\textit{iii) Level-2 behavior:}
{The level-2 players have the belief that the other players are level-1 and best respond accordingly.} Therefore, the level-2 firm $S$ will best respond to the level-1 firm $B$ and the level-2 firm $B$ will best respond to the level-1 firm $S$. We then have 
\begin{equation*}
    \begin{aligned}
        &q_S^{(2)} = \mathcal{B}_S(q_B^{(1)}) = \max\left\{ \frac{b-c}{4a}, \frac{b-c}{2a}-\frac{f}{2} \right\},\\
        &q_B^{(2)} = \mathcal{B}_B(q_S^{(1)}) = \min\left\{ \frac{b-c}{2a}+\frac{f}{4}, f \right\}.
    \end{aligned}
\end{equation*}

\textit{iv) Level-$k$ behavior:}
\ly{In general, a level-$k$ firm $i$, with $k\ge1$, best responds to a level-$(k-1)$ firm $j$:}
$$q_i^{(k)} = \mathcal{B}_i(q_j^{(k-1)}),$$
where $i,j \in \{S,B\}$ and $i \neq j$.

\section{Level-\textit{k} Social Welfare Performance}\label{sec: level k welfare performance}
\ly{In this section, we derive closed-form expressions for level-$k$ behavior and establish their properties. We then study how the firms' rationality levels affect social welfare. The next section compares this welfare with Nash equilibrium welfare.}

\begin{lemma}\label{lemma: closed-form level k behaviors}
    For any $k \in \mathbb{N}$, the level-$k$ behavior of firm $S$ is given by 
    \begin{equation*}
        \begin{aligned}
            q_S^{(k)} \!=\! 
            \begin{dcases}
                &\!\!\!\!\!\!\! \max\left\{ \frac{b-c}{2^{\frac{k}{2}+1}a}, \frac{b-c}{2a}-\frac{f}{2} \right\}, \;\text{$k$ even}; \\
                &\!\!\!\!\!\!\! \max\left\{ \frac{b-c}{2^{\frac{k+1}{2}}a} \!-\! \frac{f}{2^{\frac{k+1}{2}+1}}, \frac{b-c}{2a} \!-\! \frac{f}{2} \right\}, \;\text{$k$ odd}.
            \end{dcases}
        \end{aligned}
    \end{equation*}
    Similarly, the level-$k$ behavior of firm $B$ is given by 
    \begin{equation*}
        \begin{aligned}
            q_B^{(k)} \!=\! 
            \begin{dcases}
                &\!\!\!\!\!\!\! \min\left\{ \frac{(2^{\frac{k}{2}}-1)(b-c)}{2^{\frac{k}{2}}a} \!+\! \frac{f}{2^{\frac{k}{2}+1}}, f \right\}, \;\text{$k$ even}; \\
                &\!\!\!\!\!\!\! \min\left\{ \frac{(2^{\frac{k+1}{2}}-1)(b-c)}{2^{\frac{k+1}{2}}a}, f \right\}, \;\text{$k$ odd}.
            \end{dcases}
        \end{aligned}
    \end{equation*}
\end{lemma}
The proof is given in Appendix \ref{appendix: closed-form level k behaviors}.
\ly{The expressions for level-$k$ behavior depend on whether $k$ is even or odd. Parity-dependent behavior is also observed in duopoly markets in \cite{Runco}. This dependence does not prevent the production quantities from being monotone in $k$, as the next lemma shows.}

\begin{lemma}\label{lemma: monotonicity}
    For any $k \in \mathbb{N}$,
    \begin{itemize}
        \item $q_S^{(k)}$ is (weakly) decreasing in $k$, lower bounded by {$\frac{b-c}{2a}-\frac{f}{2}$} and upper bounded by {$\frac{b-c}{2a}$};
        \item $q_B^{(k)}$ is (weakly) increasing in $k$, lower bounded by $\frac{f}{2}$ and upper bounded by $f$.
    \end{itemize}
\end{lemma}

\begin{proof}
    \ly{We first consider even $k \in \mathbb{N}$ and write}
    \begin{equation*}
        \begin{aligned}
            q_S^{(k)} - q_S^{(k+1)} = &\max\left\{ \frac{b-c}{2^{\frac{k}{2}+1}a}, \frac{b-c}{2a}-\frac{f}{2} \right\}\\
            &- \max\left\{ \frac{b-c}{2^{\frac{k}{2}+1}a} - \frac{f}{2^{\frac{k}{2}+2}}, \frac{b-c}{2a}-\frac{f}{2} \right\}\!.
        \end{aligned}
    \end{equation*}
    \ly{Equivalently, we can write:}
    $$q_S^{(k)} - q_S^{(k+1)} = \max\Big\{ \min\left\{ x, y \right\}, \min\left\{ x-y, 0 \right\} \Big\},$$
    where 
    \begin{equation*}
        \begin{aligned}
            \begin{dcases}
                x = \frac{f}{2^{\frac{k}{2}+2}},\\
                {y = -\frac{b-c}{2a}+\frac{f}{2} + \frac{b-c}{2^{\frac{k}{2}+1}a}.}
            \end{dcases}
        \end{aligned}
    \end{equation*}
    Then, we observe the following facts: if $x \ge y$, $q_S^{(k)} - q_S^{(k+1)} = \max\{ y, 0 \} \ge 0$; otherwise, $q_S^{(k)} - q_S^{(k+1)} = \max\{ x, x-y \} \ge 0$ since $x \ge 0$. \ly{Thus, $q_S^{(k)} - q_S^{(k+1)} \ge 0$ holds for even $k$.} {When $k$ is odd, $q_S^{(k)} - q_S^{(k+1)} \ge 0$ also holds, which can be established by similar techniques.}

    Combining {both} cases, we conclude that $q_S^{(k)}$ is weakly decreasing in $k$. The upper and lower bounds of $q_S^{(k)}$ are then straightforward by letting $k=0$ and $k \to \infty$, respectively. 

    The monotonicity and bounds of $q_B^{(k)}$ can be established similarly. 
\end{proof}

\ly{As reasoning depth increases, each firm expects its opponent to reason more deeply as well. Firm $S$ expects greater output from the planner and weakly reduces its own output to maximize profit. Firm $B$ expects lower output from firm $S$ and weakly increases its own output to improve welfare, subject to its capacity.}

In a static game, two suppliers may {have} different levels of rationality. We denote the rationality levels of the self-interested firm $S$ and the benevolent social planner $B$ by $k$ and $(k+\Delta)$, respectively, where $k \in \mathbb{N}$ and $\Delta \in \mathbb{Z}_{[-k, \infty)}$\footnote{Here, the shorthand $\mathbb{Z}_{[-k, \infty)}$ is the set of all integers greater than or equal to $-k$.}. Next, we take the social planner's perspective and fix $k$ to investigate the relationship between its level of rationality and the resulting social welfare. \ly{We refer to this welfare as the \textit{level-$k$ performance}.} \ly{The next lemma expresses welfare in terms of the distance between total output and its welfare-optimal value.}

\begin{lemma}\label{lemma: equivalent social welfare measure}
    The social welfare $W(q_S, q_B)$ \ly{is strictly decreasing in {$d(q_S, q_B) = \big|(b-c)/a-q_S-q_B\big|$}.}
\end{lemma}
\begin{proof}
    With {$d(q_S, q_B) = \big|(b-c)/a-q_S-q_B\big|$}, we can rewrite social welfare as follows, 
    $$W(q_S,q_B) = -d^2(q_S, q_B) \cdot a/2 + (b-c)^2/2a + m.$$
    Hence, $W(q_S, q_B)$ \ly{is strictly decreasing in $d(q_S, q_B)$ since} $d(q_S, q_B) \ge 0$.
\end{proof}

\ly{Under level-$k$ reasoning, total output may exceed the welfare-optimal quantity, so $(b-c)/a-q_S-q_B$ may be negative. This can occur because a firm's belief about its opponent's output need not match that opponent's actual output.} Note that if $d=0$, the best possible social welfare is achieved, where the market price is equal to the cost per unit. We call it the \textit{optimal social welfare}\footnote{{This is also known as the competitive equilibrium, which can be established by the first welfare theorem \cite{microeconomic}.}}. We denote the corresponding total quantity produced as {$q_S^{SO} + q_B^{SO} = (b-c)/a$} and call it the \textit{optimal social total production}. \ly{Given a belief about firm $S$'s output, the planner chooses its own output to bring the expected total as close as possible to the welfare-optimal quantity, subject to its capacity.}

{Next}, we consider the following question: if the level of the self-interested firm is known, is it always better for the benevolent social planner to have a high level of rationality? The following theorem presents the answer.

\begin{theorem}\label{thm: monotonicity of performance}
    For any fixed $k \in \mathbb{N}$, {the level-$k$ performance} $W(q_S^{(k)}, q_B^{(k+\Delta)})$ is (weakly) increasing in $\Delta \in \mathbb{Z}_{[-k, 1]}$ and (weakly) decreasing in $\Delta \in \mathbb{Z}_{[1, \infty)}$.
\end{theorem}
\begin{proof}
    {First, social welfare} $W(q_S^{(k)}, q_B^{(k+1)})$ is maximal among all $W(q_S^{(k)}, q_B^{(k+\Delta)})$ with $\ly{\Delta \in \mathbb{Z}_{[-k,\infty)}}$, since by definition the level-$(k+1)$ firm $B$ is a social welfare maximizer and best responds to the level-$k$ firm $S$. Due to Lemma \ref{lemma: equivalent social welfare measure}, the corresponding total quantity produced should be no larger than {$(b-c)/a$}, i.e., {$q_S^{(k)}+q_B^{(k+1)} \le (b-c)/a$}. Since $q_B^{(k)}$ is increasing in $k$ by Lemma \ref{lemma: monotonicity}, we have {$d(q_S^{(k)},q_B^{(k+\Delta)})=(b-c)/a-(q_S^{(k)}+q_B^{(k+\Delta)})$} for any $\Delta \in \mathbb{Z}_{[-k,0]}$. \ly{It follows that}
    $$d(q_S^{(k)},q_B^{(k+\Delta+1)}) \le d(q_S^{(k)},q_B^{(k+\Delta)}), \;\forall \Delta \in \mathbb{Z}_{[-k,0]}.$$
    By Lemma \ref{lemma: equivalent social welfare measure}, \ly{social welfare is weakly increasing in $\Delta \in \mathbb{Z}_{[-k, 1]}$.}

    For the case of $\Delta \in \mathbb{Z}_{[1,\infty)}$, we prove by contradiction. Assume there exists a $\Delta \in \mathbb{Z}_{[1,\infty)}$ such that $W(q_S^{(k)}, q_B^{(k+\Delta+1)}) > W(q_S^{(k)}, q_B^{(k+\Delta)})$. Then, by Lemma \ref{lemma: equivalent social welfare measure}, we have $d(q_S^{(k)}, q_B^{(k+\Delta+1)}) < d(q_S^{(k)}, q_B^{(k+\Delta)})$, which implies the following by squaring both sides of the inequality,
    $$(q_B^{(k+\Delta+1)} \!-\! q_B^{(k+\Delta)}) \!\!\left(\! 2q_S^{(k)} \!+\! q_B^{(k+\Delta+1)} \!+\! q_B^{(k+\Delta)} \!-\! 2\frac{b\!-\!c}{a} \!\right)\! < 0.$$
    By the monotonicity of $q_B^{(k)}$, we apply $q_B^{(k+\Delta+1)} > q_B^{(k+\Delta)}$ to the preceding inequality and have
    \begin{equation}\label{eq: temp1}
        \left( q_S^{(k)}+q_B^{(k+\Delta+1)} \right) + \left( q_S^{(k)}+q_B^{(k+\Delta)} \right) < {2\frac{b-c}{a}},
    \end{equation}
    which indicates $q_S^{(k)}+q_B^{(k+\Delta)} < {(b-c)/a}$. Then, two cases may occur as follows,
    
    i) if $q_S^{(k)}+q_B^{(k+\Delta+1)} < {(b-c)/a}$: we observe that
    $$d(q_S^{(k)}, q_B^{(k+\Delta+1)}) < d(q_S^{(k)}, q_B^{(k+\Delta)}) \le d(q_S^{(k)}, q_B^{(k+1)}),$$
    since $q_B^{(k+\Delta+1)} > q_B^{(k+\Delta)}$ and $q_B^{(k)}$ is increasing. By Lemma \ref{lemma: equivalent social welfare measure}, we deduce that $W(q_S^{(k)}, q_B^{(k+\Delta+1)}) > W(q_S^{(k)}, q_B^{(k+1)})$. {However, this contradicts the fact that} $W(q_S^{(k)}, q_B^{(k+1)})$ is the maximum among $W(q_S^{(k)}, q_B^{(k+\Delta)})$ for any $\Delta \in \mathbb{Z}_{[-k,\infty)}$.

    ii) if $q_S^{(k)}+q_B^{(k+\Delta+1)} > {(b-c)/a}$: then the inequality \eqref{eq: temp1} shows
    $$0<\left( q_S^{(k)}+q_B^{(k+\Delta+1)} \right) \!-\! {\frac{b-c}{a}}  < {\frac{b-c}{a}} \!-\! \left( q_S^{(k)}+q_B^{(k+\Delta)} \right).$$
    Equivalently, we have
    $$d(q_S^{(k)}, q_B^{(k+\Delta+1)}) < d(q_S^{(k)}, q_B^{(k+\Delta)}) \le d(q_S^{(k)}, q_B^{(k+1)}),$$
    since $q_S^{(k)}+q_B^{(k+\Delta)} < {(b-c)/a}$ and $q_B^{(k)}$ is increasing. Similarly, it contradicts $W(q_S^{(k)}, q_B^{(k+1)})$ being maximal and so $d(q_S^{(k)}, q_B^{(k+1)})$ being minimal. Therefore, we conclude that $W(q_S^{(k)}, q_B^{(k+\Delta)})$ is weakly decreasing in $\Delta \in \mathbb{Z}_{[1,\infty)}$. 
\end{proof}

\ly{For a fixed level-$k$ opponent, welfare is weakly decreasing as the planner's level moves away from $k+1$ in either direction. Thus, deeper reasoning does not always improve the planner's payoff. Level-$(k+1)$ behavior maximizes welfare because it best responds to the opponent's actual output. Other levels may attain the same welfare, for example when the capacity constraint binds.}

\begin{corollary}
    \ly{For any fixed $k \in \mathbb{N}$, if $\Delta<1$, total output is no greater than the welfare-optimal total output. If $\Delta>1$, total output either equals that at $\Delta=1$ or exceeds the welfare-optimal total output.}
\end{corollary}
\begin{proof}
    Since by definition {the} level-$(k+1)$ firm $B$ best responds to the level-$k$ firm $S$ and maximizes social welfare, the corresponding total produced quantity is either equal to the {optimal social total production} or smaller due to flow capacity {constraints}. For $\Delta<1$, we have 
    $$q_S^{(k)}+q_B^{(k+\Delta)} \le {(b-c)/a} = q_S^{SO}+q_B^{SO},$$
    based on the monotonicity of $q_B^{(k)}$ as shown in Lemma \ref{lemma: monotonicity}. 
    
    Next, we prove the case of $\Delta>1$ by contradiction. Since $q_B^{(k)}$ is increasing in $k$, we first see that $q_S^{(k)}+q_B^{(k+\Delta)} \ge q_S^{(k)}+q_B^{(k+1)}$ always holds. Thus, we assume there exists a $\Delta>1$ such that
    $$q_S^{(k)}+q_B^{(k+1)} < q_S^{(k)}+q_B^{(k+\Delta)} \le q_S^{SO}+q_B^{SO}.$$
    By Lemma \ref{lemma: equivalent social welfare measure}, we then have $d(q_S^{(k)},q_B^{(k+\Delta)}) < d(q_S^{(k)},q_B^{(k+1)})$ and {thus} $W(q_S^{(k)},q_B^{(k+\Delta)}) > W(q_S^{(k)},q_B^{(k+1)})$. However, {this} contradicts Theorem \ref{thm: monotonicity of performance}, {which shows that $W(q_S^{(k)}, q_B^{(k+\Delta)})$ is weakly decreasing} in $\Delta \in \mathbb{Z}_{[1, \infty)}$.
\end{proof}
When the benevolent social planner is less rational than the self-interested firm, then compared to the outcome with $\Delta=1$, the consumer faces a higher market price {and} achieves a {lower} utility, but the cost of production will be reduced. On the other hand, if the benevolent social planner is more rational than the self-interested firm, then the consumer benefits from the social planner and a higher utility may be achieved. Meanwhile, the drawback is that the social planner's production cost will be increased. Hence, we observe that the planner being less rational could be beneficial to the supply side and being more rational could be beneficial to the consumer side. 

{To} achieve good social welfare performance, the benevolent social planner needs to make a tradeoff between the consumer's utility and production cost. We will further discuss different strategies for the benevolent social planner to secure or improve social welfare in {Section \ref{sec: planner strategies}}.

\section{\ly{Comparison of Outcomes under Bounded and Full Rationality}}\label{sec: level k vs NE}
Next, we compare \ly{outcomes under bounded rationality with Nash equilibrium outcomes under full rationality.} We refer to the social welfare at the Nash equilibrium  in the fully rational setting as \textit{equilibrium performance} and denote it by $W^{NE}$. \ly{We compare these outcomes using the \textit{price of rationality (PoR)}, the ratio of equilibrium welfare to level-$k$ welfare:}
$$\ly{PoR(k, k+\Delta) = \frac{W^{NE}}{W(q_S^{(k)}, q_B^{(k+\Delta)})}.}$$
{{If $PoR$ is greater than 1}, the level-$k$ performance is worse than the equilibrium performance, and thus the fully rational Nash equilibrium analysis results in an overestimated social welfare; {if $PoR$ is less than 1}, the equilibrium performance is worse, meaning that the social welfare is underestimated by the NE solution concept when the actual outcome is generated by level-$k$ players; if $PoR$ {equals 1}, the level-\textit{k} performance is modeled accurately by {the} Nash outcome, though their strategy profiles may still be different. We first characterize the Nash equilibrium and the corresponding social welfare in the following lemma.}

\begin{lemma}\label{lemma: NE strategy profile}
    The Nash equilibrium strategy profile is
    $$(q_S^{NE}, q_B^{NE}) = {\Big( (b-c)/2a-f/2, f \Big),}$$
    with the equilibrium performance $W^{NE} = W(q_S^{NE}, q_B^{NE})$.
\end{lemma}
To establish the lemma, we use the definition of Nash equilibrium, where $q_S^* = \mathcal{B}_S(q_B^*)$ and $q_B^* = \mathcal{B}_B(q_S^*)$. 

We see that the level-$k$ behaviors of both firms approach the Nash equilibrium as $k$ becomes large, i.e., $\lim_{k \to \infty} q_i^{(k)} = q_i^{NE}$ with $i\in \{S,B\}$. \ly{In this model, the recursive best responses converge to the mutually consistent Nash strategy profile. We can therefore interpret Nash equilibrium as the limit of level-$k$ behavior as both firms' reasoning depths increase.}

\begin{theorem}
    For any {$k \in \mathbb{N}$ and $\Delta \in \mathbb{Z}_{[-k,0)}$}, the price of rationality is {greater than or equal to 1}, i.e.,
    $$PoR(k, k+\Delta) \ge {1}, \quad \text{for } \forall \Delta \in \mathbb{Z}_{[-k,0)}.$$
    Equivalently, the level-$k$ performance is no better than the equilibrium performance. 
\end{theorem}
\begin{proof}
    We first let $\Delta=-1$ and prove the statement. For any $k \in \mathbb{N}$ that is even, we have 
    $$q_S^{(k)} + q_B^{(k-1)} = \max\left\{ \min\left\{ x, y \right\}, \min\left\{ z, \frac{x+y}{2} \right\} \right\},$$
    where 
    \begin{equation*}
        \begin{aligned}
            \begin{dcases}
                x = \frac{(2^{\frac{k}{2}+1}-1)(b-c)}{2^{\frac{k}{2}+1}a},\\
                y = \frac{b-c}{2^{\frac{k}{2}+1}a}+f,\\
                z = \frac{b-c}{2a}+\frac{(2^{\frac{k}{2}}-1)(b-c)}{2^{\frac{k}{2}}a} - \frac{f}{2}.
            \end{dcases}
        \end{aligned}
    \end{equation*}
    Since $\min\{ x,y \} \le (x+y)/2$ and $\min\{ z,(x+y)/2 \} \le (x+y)/2$, we have the following observation:
    $$q_S^{(k)} + q_B^{(k-1)} \le \frac{x+y}{2} = {\frac{b-c}{2a}+\frac{f}{2}}.$$
    For comparison, in the fully rational setting, the total produced quantity at Nash equilibrium is
    $$q_S^{NE} + q_B^{NE} = {\frac{b-c}{2a}+\frac{f}{2}} \ge q_S^{(k)} + q_B^{(k-1)}.$$
    With the assumption of {$f \le (b-c)/a$}, we conclude that
    $$d(q_S^{(k)}, q_B^{(k-1)}) \ge d(q_S^{NE}, q_B^{NE}),$$
    which implies $W(q_S^{(k)}, q_B^{(k-1)}) \le W^{NE}$ by Lemma \ref{lemma: equivalent social welfare measure}. {For any $k \in \mathbb{N}$ that is odd, we can also show $W(q_S^{(k)}, q_B^{(k-1)}) \le W^{NE}$ by using similar techniques.} Based on Theorem \ref{thm: monotonicity of performance}, we conclude that for any $\Delta<0$, there holds
    \begin{equation*}
        W(q_S^{(k)}, q_B^{(k+\Delta)}) \le W(q_S^{(k)}, q_B^{(k-1)}) \le W^{NE}.
    \end{equation*}
\end{proof}

This result implies that the social welfare performance is no better than the equilibrium performance, i.e., {$PoR\geq 1$}, if the benevolent social planner has a lower level of rationality than the self-interested firm. 
On the other hand, if the benevolent firm has an equal or higher level of rationality, the following theorem shows that the price of rationality can be {less than 1} under some circumstances.

\begin{theorem}\label{thm: level-k NE comparison}
    For any $k,\Delta \in \mathbb{N}$ with $k+\Delta > 0$, the price of rationality is strictly less than 1, i.e.,
    $$PoR(k, k+\Delta) < 1,$$
    if and only if the following conditions hold,
    \begin{enumerate}
        \item when $k$ is even: 
        \begin{equation*}
            \frac{f}{(b-c)/a}\in
            \begin{dcases}
                (1-\frac{1}{2^{\frac{k}{2}}}, 1), \quad\text{if } \Delta = 0,1;\\
                (1-\frac{1}{2^{\frac{k}{2}}}, 1-\frac{1}{3\cdot2^{\frac{k}{2}}}), \quad\text{if } \Delta \ge 2.
            \end{dcases}
        \end{equation*}
        \item when $k$ is odd: 
            \begin{equation*}
                \frac{f}{(b-c)/a} \!\in\!
                \begin{dcases}
                    (1 \!-\! \frac{1}{2^{\frac{k+1}{2}}-1}, 1 \!-\! \frac{1}{2^{\frac{k+1}{2}}+1}), \;\text{if } \Delta = 0;\\
                    (1 \!-\! \frac{1}{2^{\frac{k+1}{2}}-1}, 1), \;\text{if } \Delta = 1,2;\\
                    (1 \!-\! \frac{1}{2^{\frac{k+1}{2}}\!-\!1}, 1 \!-\! \frac{1}{3\!\cdot\! 2^{\frac{k+1}{2}}\!-\!1}), \;\text{if } \Delta \ge 3.
                \end{dcases}
            \end{equation*}
    \end{enumerate}
\end{theorem}
The proof is given in Appendix \ref{appendix: supplement}.

Therefore, if the rationality level of the benevolent social planner is {equal to or higher than that of the self-interested firm}, the achieved social welfare may be better than the equilibrium welfare, depending on the {production capacity}.

\section{Social Planner's Strategies}\label{sec: planner strategies}
\ly{We now allow the planner to choose any feasible output. Firm $S$ continues to follow level-$k$ reasoning and retains its belief that the planner follows that model. Thus, firm $S$'s output remains $q_S^{(k)}$. We ask how the planner should respond when it knows the opponent's exact level, knows only its probability distribution, or has no distributional information. We also ask whether the resulting strategies coincide with level-$k$ behavior.}

\subsection{Optimal strategy}
\ly{With complete information, the planner knows firm $S$'s rationality level $k$ and thus its output $q_S^{(k)}$. The planner's \textit{optimal strategy} is the best response $q_B^*=q_B^{(k+1)}$, according to Theorem \ref{thm: monotonicity of performance}. Note that even though the social planner does not necessarily need to follow the level-$k$ reasoning, its optimal strategy is still to behave as level-$(k+1)$, since it aims to maximize the social welfare and level-$(k+1)$ behavior best responds to level-$k$ behavior.}

\subsection{Expectation maximizing strategy}
In general, it may be {difficult} for the benevolent social planner to have {complete} information and know the exact rationality level of the self-interested firm. In this case, the optimal strategy cannot be achieved. {Here, following the classical Bayesian game setup, we suppose firm $S$'s level of rationality is a random variable and its probability distribution is known to the benevolent social planner.} \ly{We write $\mathbb{P}(K=k)$ for the probability} that firm $S$ is of level-$k$ and $\sum_{k\in\mathbb{N}}\mathbb{P}(K=k) = 1$. \ly{We define the \textit{stochastic strategy} as the quantity produced by firm $B$ that maximizes expected social welfare under this distribution. Here, ``stochastic'' refers to uncertainty about the opponent's level; the planner chooses a fixed output:}
$$q_B^{SS} = \underset{0 \le q_B \le f}{\arg\max} \;\;\mathbb{E} \left[ W(q_S^{(K)}, q_B) \right].$$
\ly{The next theorem gives the solution.}

\begin{theorem}
     \ly{The stochastic strategy is given by}
    $${q_B^{SS} = \min\left\{ (b-c)/a-\mathbb{E}[q_S^{(K)}], f \right\}.}$$
\end{theorem}
\begin{proof}
    {First, we write the optimization problem as follows:}
    $$\max_{0 \le q_B \le f} \mathbb{E} \left[ W(q_S^{(K)}, q_B) \right].$$
    {According to Lemma \ref{lemma: monotonicity}, $q_S^{(k)}$ is upper bounded by {$(b-c)/2a$} and lower bounded by $(b-c)/2a-f/2$ for any $k \in \mathbb{N}$. Thus, $\mathbb{E}[q_S^{(K)}] < (b-c)/a$ and $\mathbb{E}[(q_S^{(K)})^2]$ is finite. Therefore, the optimization problem is equivalent to}
    $$\max_{0 \le q_B \le f} q_B\left( \frac{b-c}{a}-\mathbb{E}[q_S^{(K)}] - \frac{1}{2}q_B \right),$$
    By the {first-order} condition and the constraint on $q_B$, {the optimal solution is}
    $$q_B^{SS} = \min\left\{ (b-c)/a-\mathbb{E}[q_S^{(K)}], f \right\}.$$
\end{proof}

\begin{remark}
    {It is worth noting that $q_B^{SS}$ may not equal $\mathbb{E}[q_B^{(K+1)}]$ due to the {production capacity constraint}.}
\end{remark}

\ly{For a realized opponent level $k$, the \textit{value of complete information (VCI)} \cite{MIT_gt} measures the welfare gain from using the complete-information strategy instead of the stochastic strategy:} %as the maximum amount of value the benevolent social planner is willing to pay for complete information prior to making a decision. In this case, we have
$$VCI^{SS} = W(q_S^{(k)}, q_B^*) - W(q_S^{(k)}, q_B^{SS}).$$
By definition, $VCI^{SS}$ is non-negative. {We will demonstrate the value of $VCI^{SS}$ by numerical studies in Section \ref{sec: numerical}.}

\subsection{Robust maximin strategy}
\ly{Next, we assume that the planner knows only that its opponent follows level-$k$ reasoning, with $k \in \mathbb{N}$. It has no probability distribution over these levels.} We define the \textit{robust strategy} to be the quantity produced by firm $B$ that maximizes the minimal social welfare among all possible rationality levels of the self-interested firm, i.e.,
\begin{equation}\label{eq: maximin problem}
    q_B^{RS} = \underset{0 \le q_B \le f}{\arg\max} \; \min_{k\in\mathbb{N}} W(q_S^{(k)}, q_B). 
\end{equation}
\ly{Surprisingly, we find that the robust strategy coincides with the planner's level-2 output, as shown below.}
\begin{theorem}
    The robust strategy is {given by} 
    $$q_B^{RS} = \min\left\{ \frac{b-c}{2a}+\frac{f}{4}, f \right\},$$
    which corresponds to its level-2 behavior. 
\end{theorem}

\begin{proof}
    \ly{To solve the maximin problem, we first consider the minimization problem:}
    \begin{equation}\label{eq: minimization problem}
        k^* = \underset{k \in \mathbb{N}}{\arg\min} \; W(q_S^{(k)}, q_B),
    \end{equation}
    with $0 \le q_B \le f$. Since $W(q_S^{(k)}, q_B)$ is a quadratic function of $q_S^{(k)}$ for any given $q_B$, the preceding problem is then equivalent to
    $$k^* = \underset{k \in \mathbb{N}}{\arg\min} \; -\left( q_S^{(k)} - \left(\frac{b-c}{a}-q_B\right) \right)^2.$$
    By Lemma \ref{lemma: monotonicity}, we have $q_S^{(k)}$ being monotonic and tightly bounded by $\frac{b-c}{2a}-\frac{f}{2}$ and $\frac{b-c}{2a}$. {Since the objective function is a concave function on $q_S^{(k)}$ parameterized by $q_B$, the optimal solution is the farthest point away from $\frac{b-c}{a}-q_B$.} Therefore, the optimal solution of the minimization problem $k^*$ satisfies
    \begin{equation*}
        q_S^{(k^*)} = 
        \begin{dcases}
            \frac{b-c}{2a}-\frac{f}{2}, \quad\text{if } q_B < \frac{b-c}{2a}+\frac{f}{4};\\
            \frac{b-c}{2a}, \quad\text{if } q_B > \frac{b-c}{2a}+\frac{f}{4};\\
            \frac{b-c}{2a}-\frac{f}{2} \text{ or } \frac{b-c}{2a}, \quad\text{if } q_B = \frac{b-c}{2a}+\frac{f}{4}.
        \end{dcases}
    \end{equation*}
    The corresponding objective value of problem \eqref{eq: minimization problem} is then $W(q_S^{(k^*)}, q_B)$. \ly{We then maximize this worst-case welfare over $0 \le q_B \le f$. Equivalently,}
    \begin{equation*}
        \max_{0 \le q_B \le f} \; -\frac{1}{2}q_B\left( q_B-2\left(\frac{b-c}{a}-q_S^{(k^*)}\right) \right).
    \end{equation*}
    Note that $q_S^{(k^*)}$ is a piecewise function in $q_B$. Thus, we can decompose the preceding problem into two subproblems. 
    
    \ly{i) Subproblem 1: substituting $q_S^{(k^*)} = \frac{b-c}{2a}-\frac{f}{2}$ gives}
    \begin{equation*}
        \max_{0 \le q_B \le \min\{\frac{b-c}{2a}+\frac{f}{4},f\}} \; -\frac{1}{2}q_B\left( q_B-\left(\frac{b-c}{a}+f\right) \right).
    \end{equation*}
    The optimal solution is then $q_B^{RS} = \min\{\frac{b-c}{2a}+\frac{f}{4},f\}$.

    \ly{ii) Subproblem 2: substituting $q_S^{(k^*)} = \frac{b-c}{2a}$ gives}
    \begin{equation*}
        \max_{\frac{b-c}{2a}+\frac{f}{4} \le q_B \le f} \; -\frac{1}{2}q_B\left( q_B-\frac{b-c}{a} \right).
    \end{equation*}
    If $\frac{b-c}{2a}+\frac{f}{4} \le f$, the optimal solution is then $q_B^{RS} = \frac{b-c}{2a}+\frac{f}{4} = \min\{\frac{b-c}{2a}+\frac{f}{4},f\}$; otherwise, there is no feasible solution. 

    Therefore, by combining the solutions of two subproblems, we conclude that the optimal solution of the maximin problem \eqref{eq: maximin problem} is $q_B^{RS} = \min\{\frac{b-c}{2a}+\frac{f}{4},f\}$, which is the same as the level-2 behavior of the benevolent social planner as presented in Lemma \ref{lemma: closed-form level k behaviors}. 
\end{proof}

Assuming the self-interested firm is level-$k$ for some $k \in \mathbb{N}$, the robust strategy maximizes the minimal social welfare, or say, optimizes the worst-case scenario. Essentially, it suggests that the benevolent social planner should have a relatively low rationality level if it does not know how rational its opponent is. It is worth noting that although the robust strategy for firm $B$ is the same as its level-2 behavior, the worst-case scenario is not when firm $S$ is level-1. Instead, it indicates that when firm $B$ takes the level-2 behavior, the capability of firm $S$ hurting social welfare is the least among all strategies firm $B$ can take.

\ly{For a realized opponent level $k$, the welfare gain from complete information relative to the robust strategy is}
\begin{equation*}
    VCI^{RS} = W(q_S^{(k)}, q_B^*) - W(q_S^{(k)}, q_B^{RS}).
\end{equation*}
\ly{The \textit{expected value of incomplete information (EVII)} measures the expected welfare gain from knowing the distribution of $K$ and using the stochastic strategy rather than the robust strategy:}
\begin{equation*}
    EVII^{RS} = \mathbb{E} \Big[W(q_S^{(K)}, q_B^{SS})\Big] - \mathbb{E} \Big[W(q_S^{(K)}, q_B^{RS})\Big].
\end{equation*}
{By definition,} both $VCI^{RS}$ and $EVII^{RS}$ are non-negative. {We will demonstrate the values of $VCI^{RS}$ and $EVII^{RS}$ by numerical studies in Section \ref{sec: numerical}.}

\section{Utility Design}\label{sec: utility design}
{Up until} this point, we assume the benevolent social planner maximizes its true payoff function, namely social welfare $W$. While it is a reasonable setup, we hypothesize that the social planner might achieve better social welfare by carefully designing its utility function. This is possible in a level-$k$ setup where agents may have conflicting beliefs about each other's rationality level, however it is not possible in NE setting, due to the self-consistent beliefs at the equilibrium. In this section, we explore how the social planner might improve social welfare by either cooperating with or fighting the self-interested firm. 

Recall that the suppliers' true payoff functions are profit $\pi_S(q_S, q_B)$ and social welfare $\pi_B(q_S, q_B)=W(q_S,q_B)$, respectively. Instead of maximizing the true payoff, we now consider the social planner maximizing a designed \textit{utility function} $U(q_S,q_B)$, i.e., $\max_{0 \le q_B \le f} U(q_S,q_B)$. While $U$ can take arbitrary forms, we consider a scenario where the social planner either cooperates with or fights the self-interested firm by maximizing a linear combination of its true payoff and its opponent's payoff, i.e.,
\begin{equation*}
    U(q_S,q_B) = \pi_B(q_S,q_B) + \gamma \pi_S(q_S,q_B).
\end{equation*}
Here, $\gamma \in \mathbb{R}$ is a design parameter, referred to as the \textit{cooperation level}. It indicates how much the social planner cooperates with or fights the self-interested firm. If $\gamma=0$, the social planner maximizes its true payoff, which has been discussed in previous sections. If $\gamma \to \infty$, the social planner fully cooperates with the self-interested firm, aligning its utility function perfectly with firm $S$'s payoff. If $\gamma \to -\infty$, the social planner aims to  minimize firm $S$'s payoff. 

The self-interested firm, on the other hand, remains a profit maximizer and acts strategically according to the best response $\mathcal{B}_S(q_B)$ as presented in Lemma \ref{lemma: best response functions}. Under this framework, we denote the produced quantity of the level-$k$ firm $i$ by $\tilde{q}_i^{(k)}$, with $i \in \{S, B\}$ and $k \in \mathbb{N}$. \ly{We retain the same mean quantities for level-0 behavior: $\tilde{q}_S^{(0)} = (b-c)/2a$ and $\tilde{q}_B^{(0)} = f/2$.}

\ly{For the rest of this section, we assume that the two suppliers have the same rationality level $k \in \mathbb{N}$.} \ly{For $k\ge1$, each firm therefore assumes that its opponent is level-$(k-1)$ and best responds to that belief, although the opponent's actual level is $k$.} To maximize social welfare $W$, the benevolent social planner may design the cooperation level $\gamma$ in the following manner,
\begin{equation}\label{eq: max problem utility design}
    \gamma^* = \underset{\gamma \in \mathbb{R}}{\arg\max} \;W(\tilde{q}_S^{(k)},\tilde{q}_B^{(k)}), \quad\forall k \in \mathbb{N},
\end{equation}
\ly{where $\tilde{q}_S^{(k)}$ and $\tilde{q}_B^{(k)}$ are the level-$k$ outputs when firm $B$ uses $U$ and firm $S$ continues to maximize profit.}
\ly{The next theorem gives an optimal cooperation level for each positive $k$.}

\begin{theorem}\label{thm: optimal cooperation level}
    {An optimal cooperation level is given by:}
    \begin{enumerate}[label={(\roman*)}]
        \item For any positive and even number $k$:
            \begin{equation*}
                \gamma^* = 2\left( \frac{(b-c)/a-f}{(b-c)/a-f/2} \right)^{\frac{2}{k}}-1.
            \end{equation*}
        \item For any positive and odd number $k$:
            \begin{equation*}
                \gamma^* = \max\left\{ -\frac{f/2}{(b-c)/a}, 2\left( \frac{(b-c)/a-f}{(b-c)/a} \right)^{\frac{2}{k+1}}\!\!-\!1 \right\}.
            \end{equation*}
    \end{enumerate}
\end{theorem}
\begin{proof}
\ly{We give a proof sketch here. The full proofs of parts (i) and (ii) are given in Appendices \ref{appendix: optimal cooperation level even} and \ref{appendix: optimal cooperation level}, respectively.} \ly{First, we divide the analysis into three cases:} $\gamma < -1$, $-1 \le \gamma \le 1$ and $\gamma > 1$. \ly{For each case,} we derive the closed-form solutions of $\tilde{q}_S^{(k)}$ and $\tilde{q}_B^{(k)}$, which both depend on the parity of $k$. Hence, we further break each case into two sub-cases: when $k$ is even and when $k$ is odd.

According to Lemma \ref{lemma: equivalent social welfare measure}, solving the maximization problem in \eqref{eq: max problem utility design} {is equivalent to solving}
\begin{equation*}
    \min_{\gamma \in \mathbb{R}} \; d(\tilde{q}_S^{(k)},\tilde{q}_B^{(k)}) = \big| (b-c)/a - (\tilde{q}_S^{(k)}+\tilde{q}_B^{(k)}) \big|.
\end{equation*}
Therefore, our goal is to find a $\gamma$ such that the total production level $(\tilde{q}_S^{(k)}+\tilde{q}_B^{(k)})$ is as close to the optimal social total production $(b-c)/a$ as possible. To this end, we write $(\tilde{q}_S^{(k)}+\tilde{q}_B^{(k)})$ {as a piecewise function} of $\gamma$, denoted by $\tilde{q}_{total}^{(k)}(\gamma)$. Then, in the domain of $\gamma \in \mathbb{R}$, there exists at least one $\gamma^*$ such that $\tilde{q}_{total}^{(k)}(\gamma^*)$ is the closest point to $(b-c)/a$. This $\gamma^*$ is then an optimal cooperation level.
\end{proof}

\begin{remark}
    Note that there may exist multiple optimal cooperation levels. For example, if $k$ is odd and the {social planner's production capacity} $f$ satisfies 
    \begin{equation}\label{eq: condition of multiple optimum}
        \left( \frac{(b-c)/a-f}{(b-c)/a-f/2} \right)^{\frac{2}{k-1}} \ge \left( \frac{(b-c)/a-f}{(b-c)/a} \right)^{\frac{2}{k+1}},
    \end{equation}
    then any $\gamma \in \mathbb{R}$ such that $\gamma \le -1$ is also optimal. The condition \eqref{eq: condition of multiple optimum} essentially means that $f$ is not very large. The intuition is as follows. When the {production capacity} is limited, the social planner only has limited market power. {Further reductions in the cooperation weight do not change welfare when the planner fights} %Thus, the actual production outputs of two suppliers are not affected by how the social planner weighs on fighting
    {the self-interested firm in this regime: all weights $\gamma\le-1$ yield the same welfare.}
\end{remark}

\ly{The next corollary gives bounds on the optimal cooperation level provided in Theorem \ref{thm: optimal cooperation level}.}
\begin{corollary}\label{cor: bound of optimal cooperation level}
    For any $k \in \mathbb{N}^+$, the optimal cooperation level proposed in Theorem \ref{thm: optimal cooperation level} satisfies $\gamma^* \in [-1,1]$. 
\end{corollary}
\begin{proof}
    The statement is established by the condition $f \le (b-c)/a$.
\end{proof}

To achieve the maximal level-$k$ performance, the theorem and corollary indicate that the social planner should primarily maximize its true payoff (social welfare), and meanwhile cooperate with or fight its opponent. Next, we investigate how well this level-$k$ performance is with cooperation or fight, by evaluating the price of rationality.

\ly{At Nash equilibrium, $\gamma=0$ achieves the highest welfare within this utility family. Cooperation or opposition therefore cannot improve on that equilibrium welfare. We retain the $\gamma=0$ Nash outcome in Lemma \ref{lemma: NE strategy profile} as the benchmark. Under level-$k$ reasoning, however, the designed utility function can improve welfare. With the optimal cooperation level, welfare is no worse than the Nash benchmark, as the next theorem shows.}

\begin{theorem}\label{thm: utility design PoR}
    With the optimal cooperation level $\gamma^*$ as stated in Theorem \ref{thm: optimal cooperation level}, the price of rationality is always less than or equal to 1, i.e., 
    \begin{equation*}
        PoR(k,k) \le 1, \quad\forall k \in \mathbb{N}.
    \end{equation*}
    Equivalently, the level-$k$ performance is no worse than the equilibrium performance. 
\end{theorem}
The proof is given in Appendix \ref{appendix: utility design PoR}.

\ly{For firms with the same reasoning level, the result shows how an appropriate cooperation level can give welfare at least as high as the Nash benchmark. It also offers a possible rationale for cooperation between a planner and private firms, as illustrated in \cite{news2}. The model does not imply that cooperation always improves welfare: the optimal weight depends on capacity and reasoning level and may be negative or zero.}

\section{Numerical Studies}\label{sec: numerical}

\ly{We use numerical experiments to compare level-$k$ welfare with the Nash benchmark, study the value of information, and illustrate utility design.} \ly{Throughout this section, we set $a=b=1$, $c=0.25$, and $m=0$. The planner's capacity satisfies $0\le f\le0.75$, as required by the model.}

\ly{First, we fix firm $S$'s rationality level at $k=1$ and compare planner levels $k+\Delta$ as capacity $f$ varies. Figure \ref{fig: PoR-production level} shows the planner's output. At small capacities, all positive planner levels produce at capacity, while level 0 ($\Delta=-1$) produces $f/2$. At larger capacities, deeper reasoning can increase the planner's output. This does not always improve welfare: a planner above level 2 can underestimate firm $S$'s actual output and therefore produce too much. Figure \ref{fig: PoR-PoR} shows the resulting price of rationality. Recall that $PoR<1$ means that level-$k$ welfare exceeds the Nash benchmark, while $PoR>1$ means that it falls below the benchmark.} %Furthermore, for any given flow capacity $f$, the level-$k$ performance is decreasing with $PoR$.
\ly{When $\Delta=-1$, the planner uses level-0 output and welfare is no better than the Nash benchmark. When $\Delta=0$, welfare may be above or below the benchmark, depending on capacity. For each capacity, $\Delta=1$ maximizes welfare over the planner's levels because the planner best responds to firm $S$'s actual output. Other levels can attain the same welfare when their outputs coincide. Comparing $\Delta=2$ with the limit $\Delta\to\infty$ shows that further reasoning can reduce welfare at larger capacities.}

\begin{figure}[ht]
    \centering
    \subfigure[Firm $B$'s production level]{
        \begin{minipage}[t]{0.5\linewidth}
        \centering
        \includegraphics[width=\linewidth]{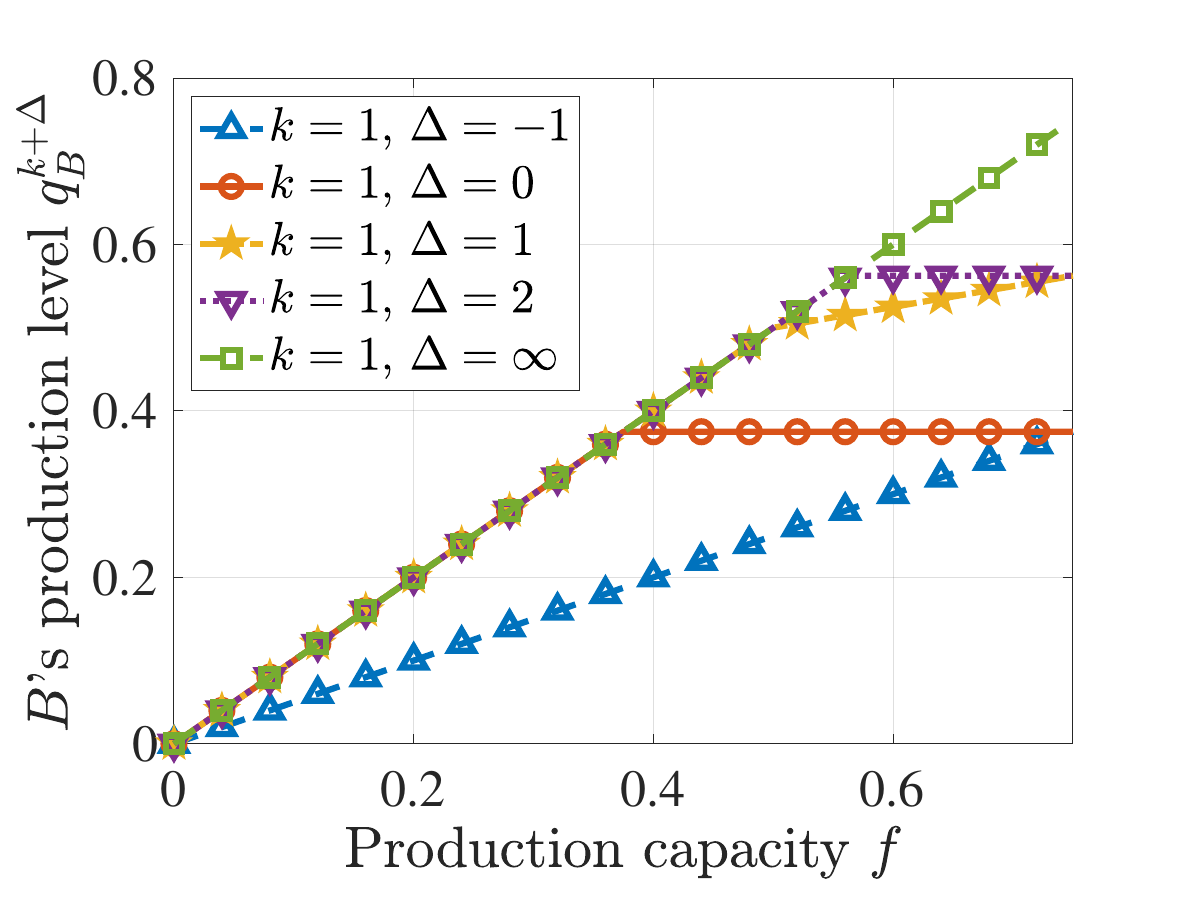}
        \label{fig: PoR-production level}
    \end{minipage}
    }%
    \subfigure[Price of rationality]{
    \begin{minipage}[t]{0.5\linewidth}
        \centering
        \includegraphics[width=\linewidth]{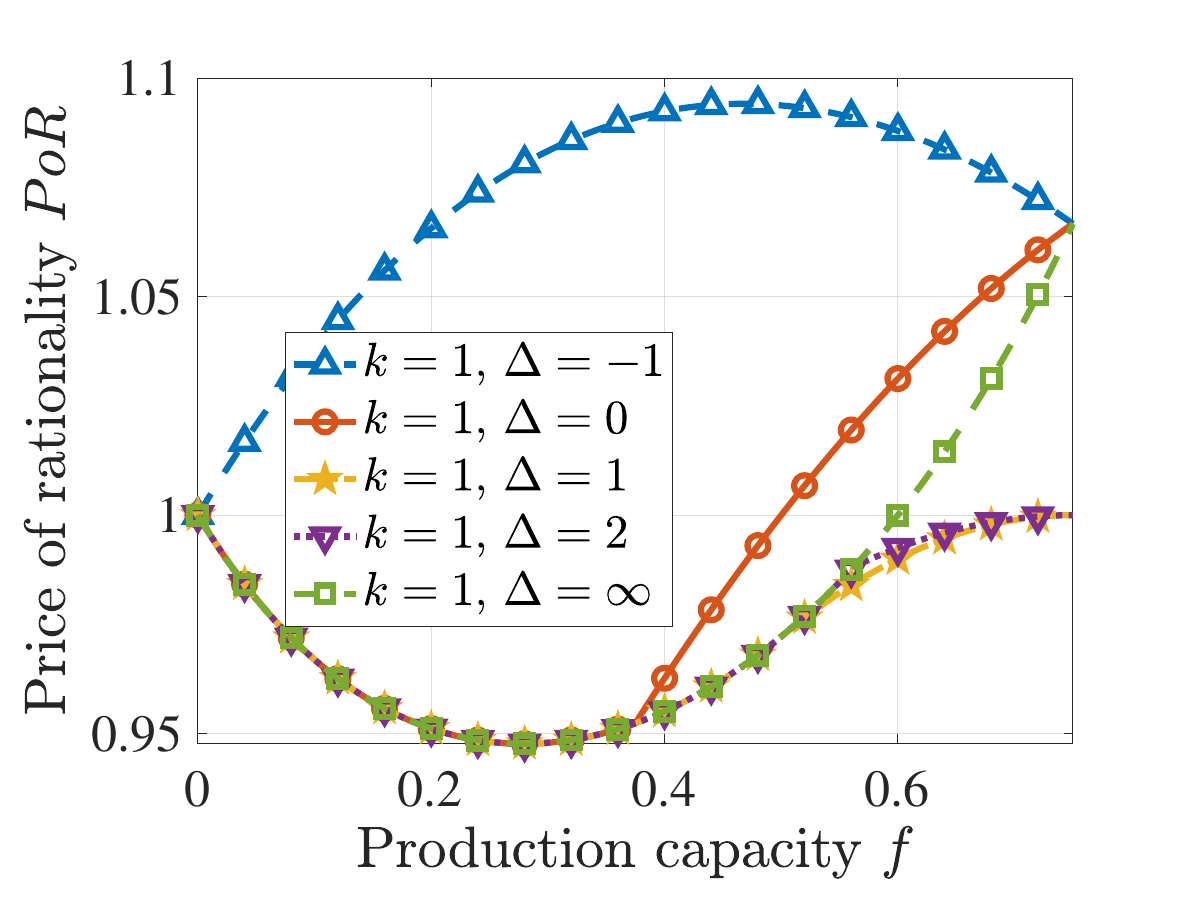}
        \label{fig: PoR-PoR}
    \end{minipage}
    }
    \caption{\ly{Comparison of level-\textit{k} welfare and equilibrium welfare.}}
    \label{fig: PoR}
\end{figure}

\ly{Next, we fix capacity and vary the planner's rationality level, to study how social welfare is affected by planner's reasoning depth. Theorem \ref{thm: monotonicity of performance} shows that welfare is weakly decreasing as $\Delta$ moves away from 1 in either direction. However, it is still unknown if it is better to have a more or less rational social planner. In Figure \ref{fig: k4}, we assume $k=4$, $\Delta=1$ maximizes welfare at each capacity. At $f=0.3$ and $f=0.5$, several other levels attain the same maximum because the planner produces at capacity. At $f=0.7$, the social welfare drops when $\Delta>1$. These cases illustrate the competing effects of capacity. A small capacity can make output insensitive to the planner's reasoning depth, while a larger capacity allows both beneficial welfare increase as well as suboptimal welfare if the planner not being sophisticated with the correct reasoning depth. Figure \ref{fig: k8} shows a wider range of relative gaps that attain maximal welfare at $k=8$. Since the horizontal axis is $\Delta$, this comparison concerns relative gaps; the planner's absolute reasoning level is $k+\Delta$.}

\begin{figure}[ht]
    \centering
    \subfigure[$k=4$]{
        \begin{minipage}[t]{0.5\linewidth}
        \centering
        \includegraphics[width=\linewidth]{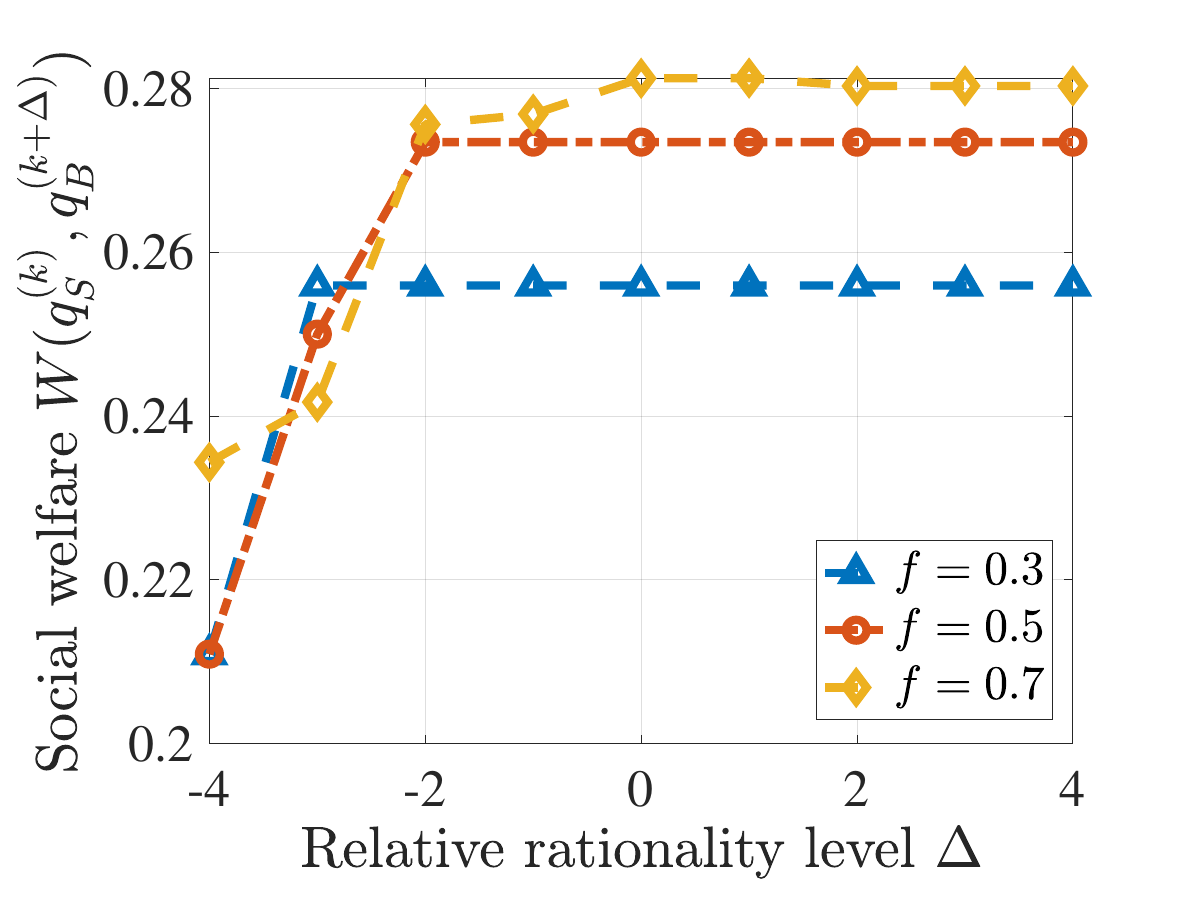}
        \label{fig: k4}
    \end{minipage}
    }%
    \subfigure[$k=8$]{
    \begin{minipage}[t]{0.5\linewidth}
        \centering
        \includegraphics[width=\linewidth]{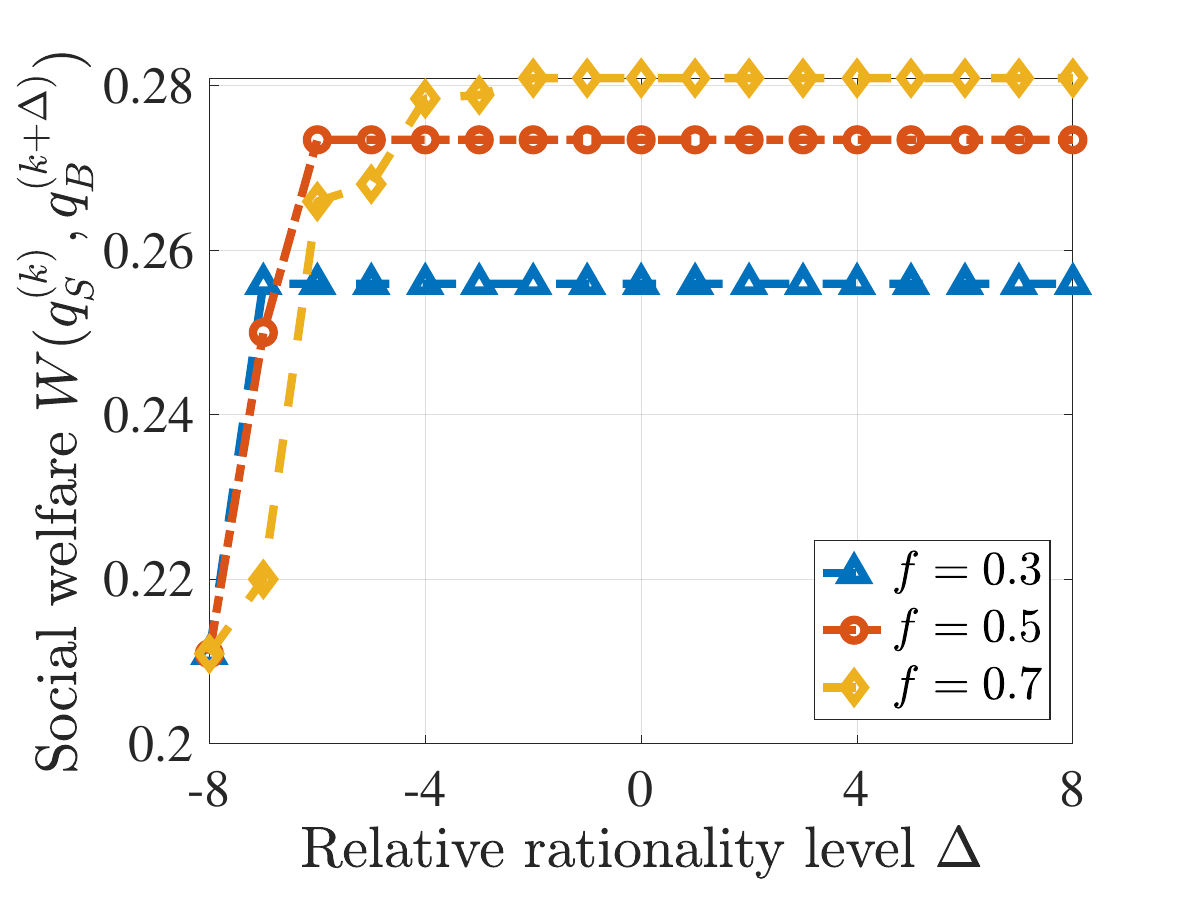}
        \label{fig: k8}
    \end{minipage}
    }
    \caption{\ly{Effect of the social planner's rationality level on social welfare.}}
    \label{fig: social welfare vs delta}
\end{figure}

\ly{We next fix the sum of the firms' rationality levels at 7, so that $k+(k+\Delta)=7$. This treats total reasoning depth as a resource to be allocated between the firms. We vary $k$ from 1 to 6, giving $\Delta\in\{-5,-3,-1,1,3,5\}$. Both firms' levels now change, whereas Theorem \ref{thm: monotonicity of performance} holds firm $S$'s level fixed. Figure \ref{fig: summation} shows the resulting welfare. First, compare allocations at a fixed capacity. At $f=0.7$, the allocation $k=3$, $\Delta=1$ attains the welfare-optimal total output: the planner correctly anticipates firm $S$'s output and has enough capacity to supply the remaining quantity. At $f=0.3$ or $f=0.5$, welfare is weakly increasing in $\Delta$. When the planner's capacity is limited, a less rational firm $S$ can help by producing more. Thus, the allocation $k=1$, $\Delta=5$ yields higher welfare than $k=3$, $\Delta=1$ in these cases, even though the planner's belief is incorrect. Next, compare capacities for a fixed allocation. At $k=6$, $\Delta=-5$, welfare is lowest among the tested allocations at each capacity. At $k=3$, $\Delta=1$, greater capacity lets the planner bring total output closer to its welfare-optimal value. At $k=1$, $\Delta=5$, increasing capacity from 0.5 to 0.7 reduces welfare: the planner underestimates firm $S$'s output, and its larger capacity permits overproduction.}

\begin{figure}[ht]
        \centering
        \begin{minipage}[b]{0.49\linewidth}
        % \centering
        \includegraphics[width=\textwidth]{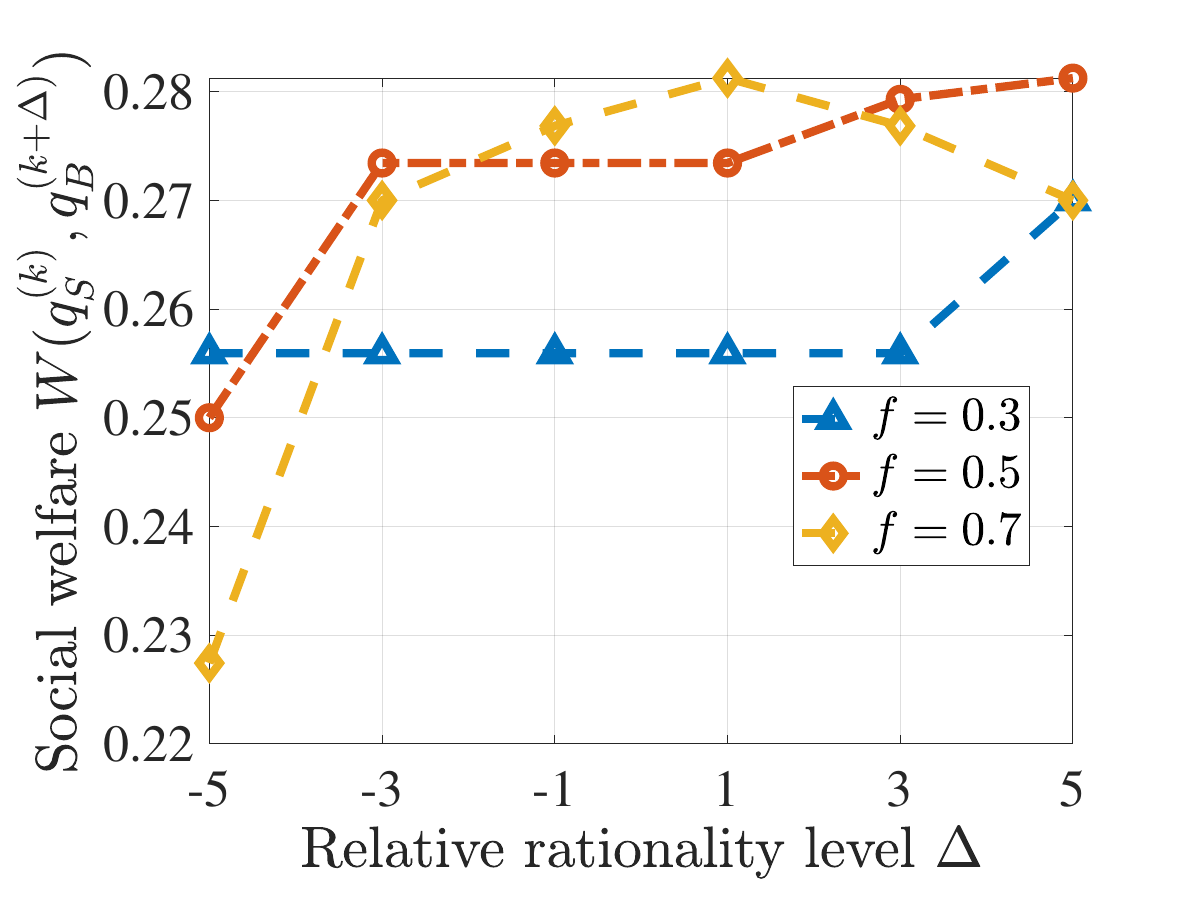}
        \caption{\ly{Social welfare when the sum of the firms' rationality levels is fixed.}}
        \label{fig: summation}
        \end{minipage}
        % \hspace{0.1cm}
        \begin{minipage}[b]{0.49\linewidth}
        % \centering
        \includegraphics[width=\textwidth]{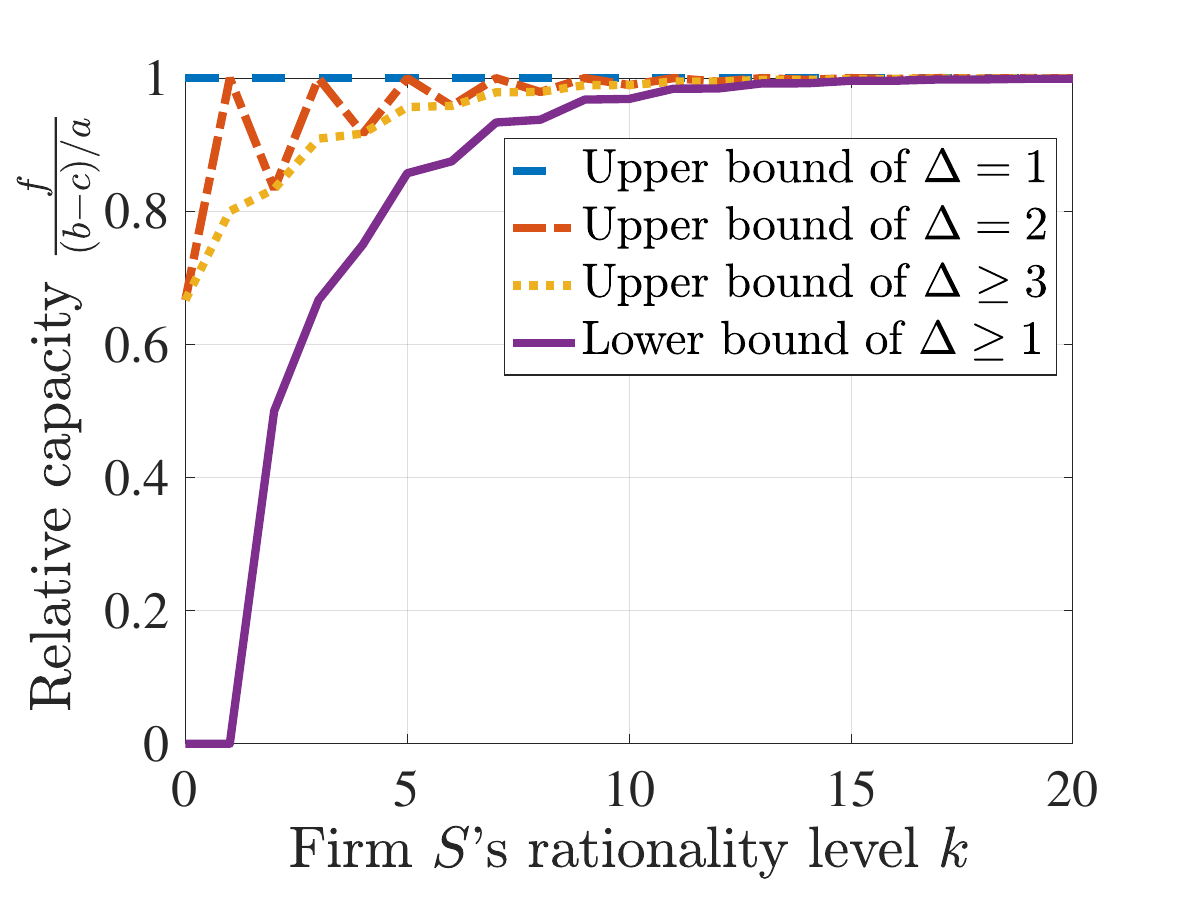}
        \caption{\ly{Region where $PoR<1$, i.e., level-$k$ welfare exceeds equilibrium welfare.}}
        \label{fig: iff region}
        \end{minipage}
        % \vspace{-0.2cm}
        \end{figure}

\ly{We then use Theorem \ref{thm: level-k NE comparison} to plot the region where $PoR<1$ for each positive relative gap $\Delta$. Figure \ref{fig: iff region} shows the bounds in terms of $k$ and normalized capacity $f/((b-c)/a)$. The dashed, dash-dotted, and dotted curves give the upper bounds for $\Delta=1$, $\Delta=2$, and $\Delta\ge3$, respectively. The solid curve gives the common lower bound. Level-$k$ welfare strictly exceeds the Nash benchmark between the corresponding bounds; the boundary points are excluded. The region shrinks as $\Delta$ increases from 1 to 3 and is unchanged for $\Delta\ge3$. This is consistent with welfare being maximal at $\Delta=1$ and weakly decreasing thereafter for fixed $k$ and $f$. The bounds vary with the parity of $k$, as do the closed-form outputs. The comparison also depends on both capacity bounds: sufficient capacity lets the planner offset firm $S$'s output, but excessive output at higher planner levels can eliminate the welfare advantage.}

\ly{Next, we compare the planner's complete-information, stochastic, and robust strategies from Section \ref{sec: planner strategies}. Following \cite{CH_Model}, we assume that firm $S$'s level $K$ follows a Poisson distribution with mean $\tau=1.5$, as shown in Figure \ref{fig: k distribution}.} %The complete information is the exact rationality level of the self-interested firm, i.e., $k$.
\ly{Figure \ref{fig: value of information} shows $VCI^{SS}$ and $VCI^{RS}$ for realized opponent levels $K=1,2,3$, together with $EVII^{RS}$, which averages over the distribution of $K$. First, the stochastic and robust strategies give similar welfare in this setting. The distribution places most probability on levels 0, 1, and 2, and the expected-welfare response is close to the robust output. The small $EVII^{RS}$ shows that knowing this distribution gives little expected welfare gain over the robust strategy. Second, for a given realized level, $VCI^{RS}$ can be smaller or larger than $VCI^{SS}$. The stochastic strategy maximizes expected welfare, but the robust strategy can perform better at a particular realization. Third, at small capacities, all three strategies produce at capacity, so information has no value in these cases. These observations concern the chosen distribution; a different distribution can change the stochastic strategy and the value of information.}

\begin{figure}[ht]
\centering
    \subfigure[Level distribution]{
        \begin{minipage}[t]{0.49\linewidth}
        \centering
        \includegraphics[width=0.97\linewidth]{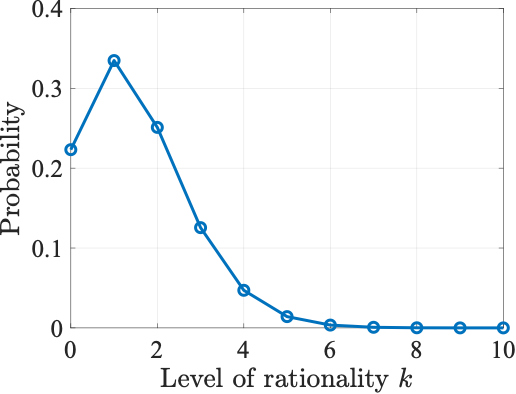}
        \label{fig: k distribution}
    \end{minipage}
    }%
    \subfigure[Value of information]{
    \begin{minipage}[t]{0.49\linewidth}
        \centering
        \includegraphics[width=0.93\linewidth]{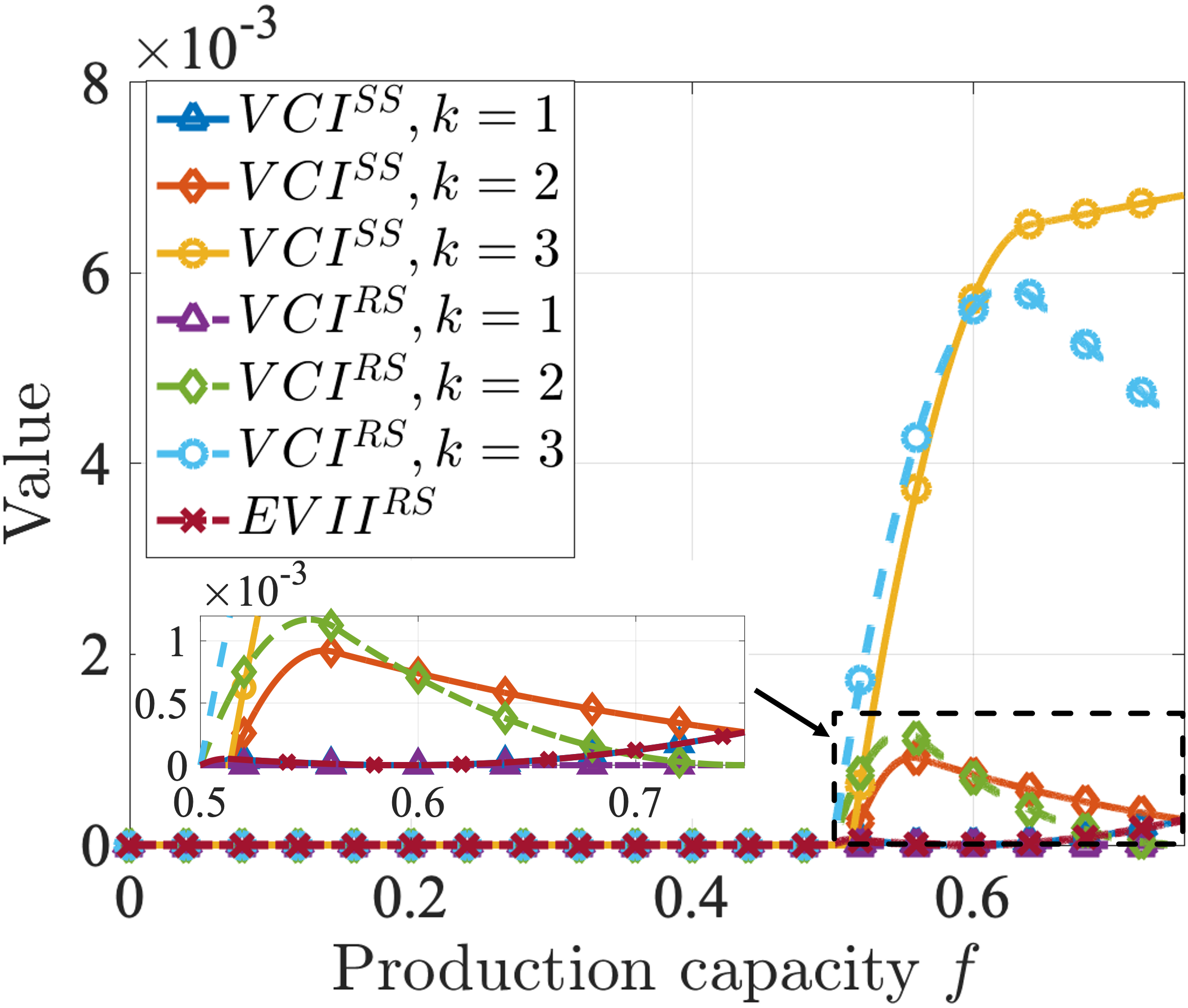}
        \label{fig: value of information}
    \end{minipage}
    }
    \caption{\ly{Comparison of the social planner's strategies under different information.}}
    \label{fig: strategy comparison}
\end{figure}
\ly{We next study utility design when both firms have the same positive rationality level $k$. Figure \ref{fig: optimal cooperation level} shows the optimal weight $\gamma^*$ from Theorem \ref{thm: optimal cooperation level}. For both odd and even $k$, this weight decreases with capacity. At low capacities, a positive weight represents cooperation with firm $S$; at higher capacities, a negative weight represents opposition. The planner selects this weight to improve realized welfare through its own output and firm $S$'s beliefs about the planner's lower-level behavior. At fixed capacity, the displayed weight is weakly increasing in $k$ within each parity class. Figure \ref{fig: utility design-PoR} shows the resulting price of rationality. Welfare is no worse than the Nash benchmark, consistent with Theorem \ref{thm: utility design PoR}. The dependence on reasoning depth differs by parity: for the even levels shown, both PoR and welfare are independent of $k$, whereas welfare is weakly decreasing along the odd levels, as Figure \ref{fig: utility design-odd case} illustrates. Thus, deeper reasoning by both firms can reduce the welfare achievable through this utility family in the odd-level cases. A fuller behavioral explanation of this parity dependence remains a topic for future work.}
\begin{figure}[ht]
\centering
    \subfigure[When $k$ is odd]{
        \begin{minipage}[t]{0.5\linewidth}
        \centering
        \includegraphics[width=\linewidth]{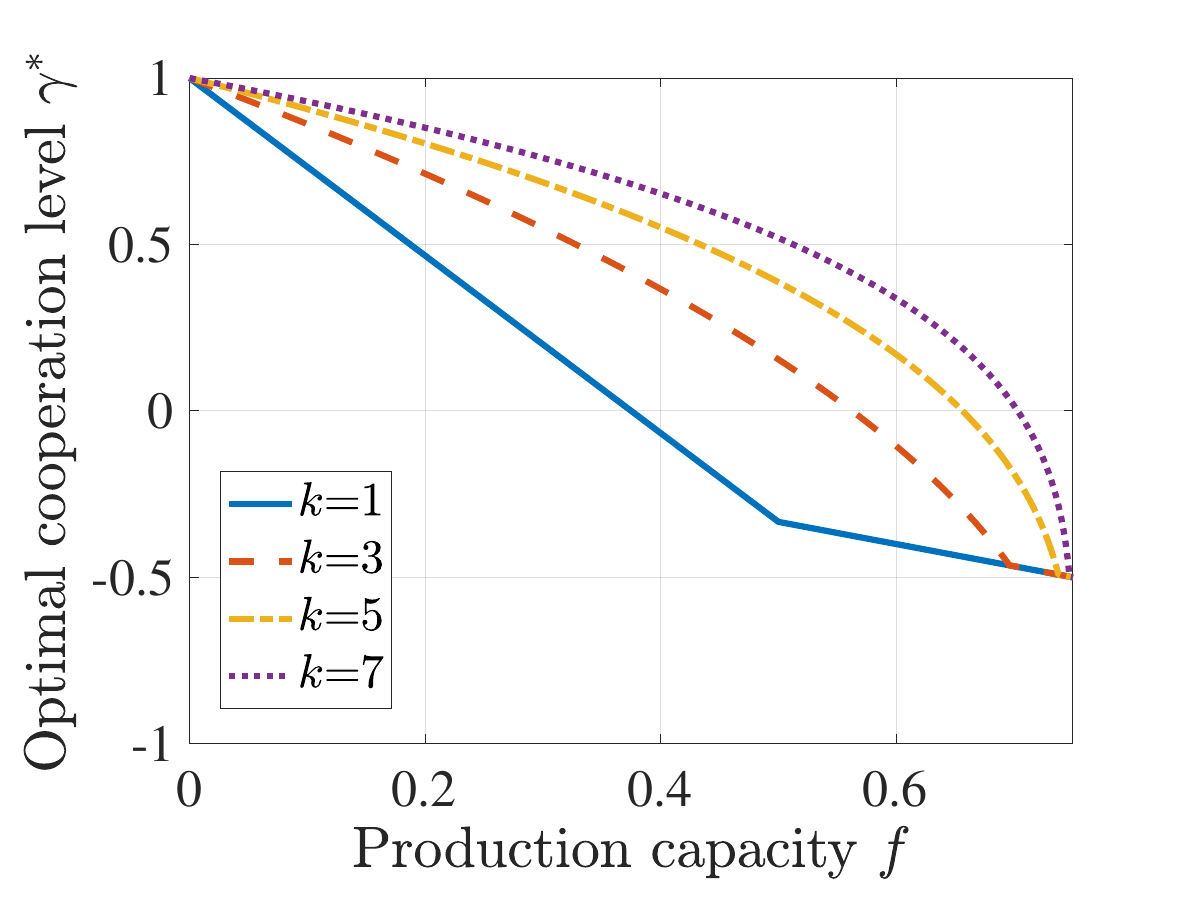}
        \label{fig: gamma-odd}
    \end{minipage}
    }%
    \subfigure[When $k$ is even]{
    \begin{minipage}[t]{0.5\linewidth}
        \centering
        \includegraphics[width=\linewidth]{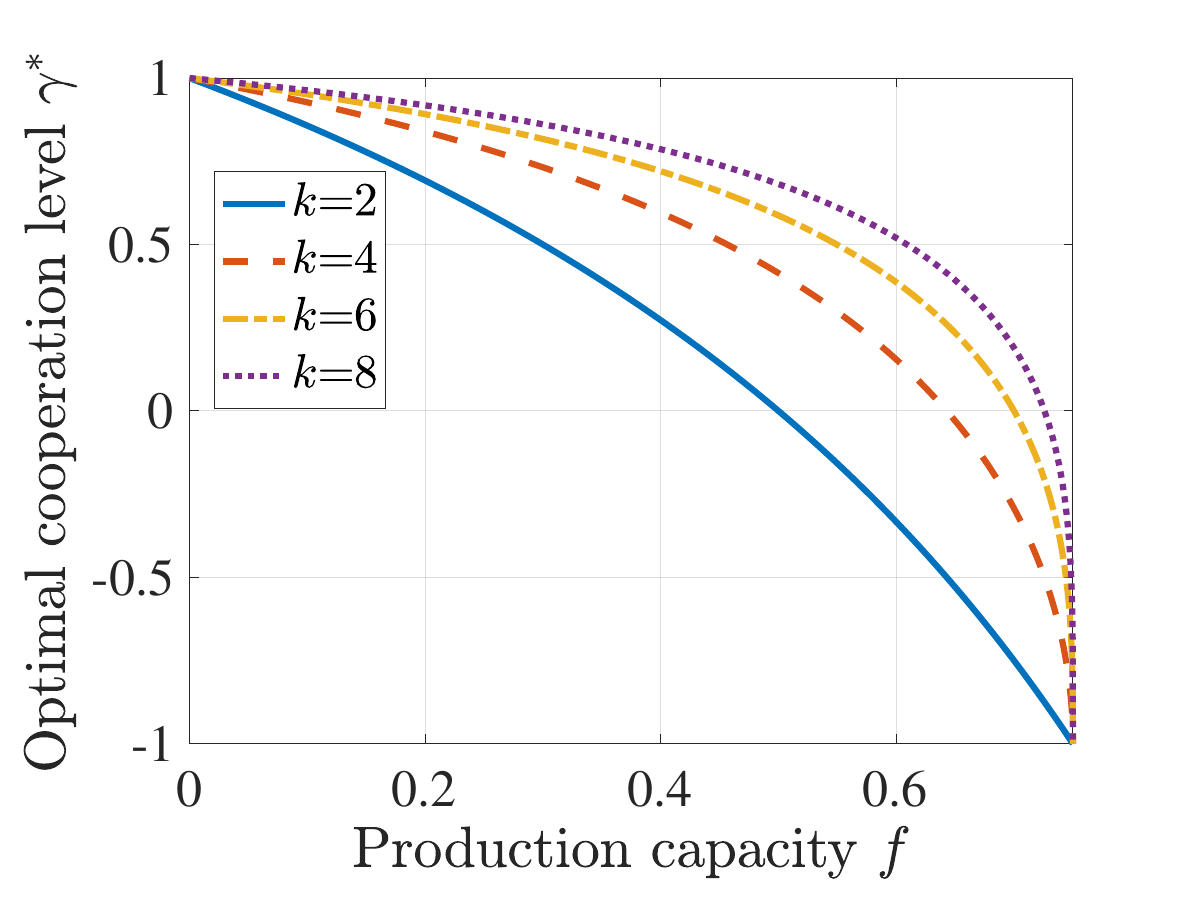}
        \label{fig: gamma-even}
    \end{minipage}
    }
    \caption{The optimal cooperation level for different rationality levels.}
    \label{fig: optimal cooperation level}
\end{figure}

\begin{figure}[ht]
\centering
    \subfigure[$PoR$ for any $k$]{
        \begin{minipage}[t]{0.48\linewidth}
        \centering
        \includegraphics[width=\linewidth]{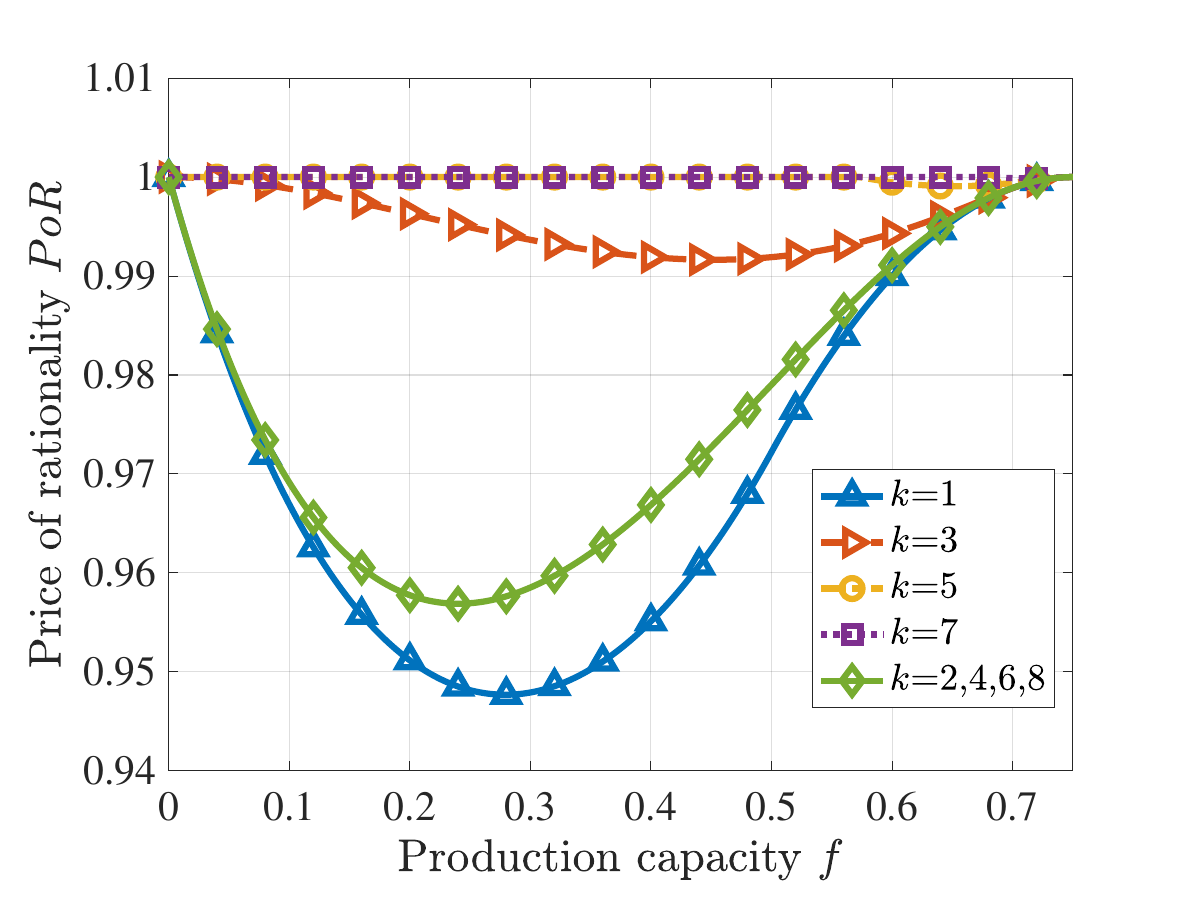}
        \label{fig: utility design-PoR}
    \end{minipage}
    }%
    \subfigure[Level-$k$ performance for odd $k$]{
    \begin{minipage}[t]{0.48\linewidth}
        \centering
        \includegraphics[width=\linewidth]{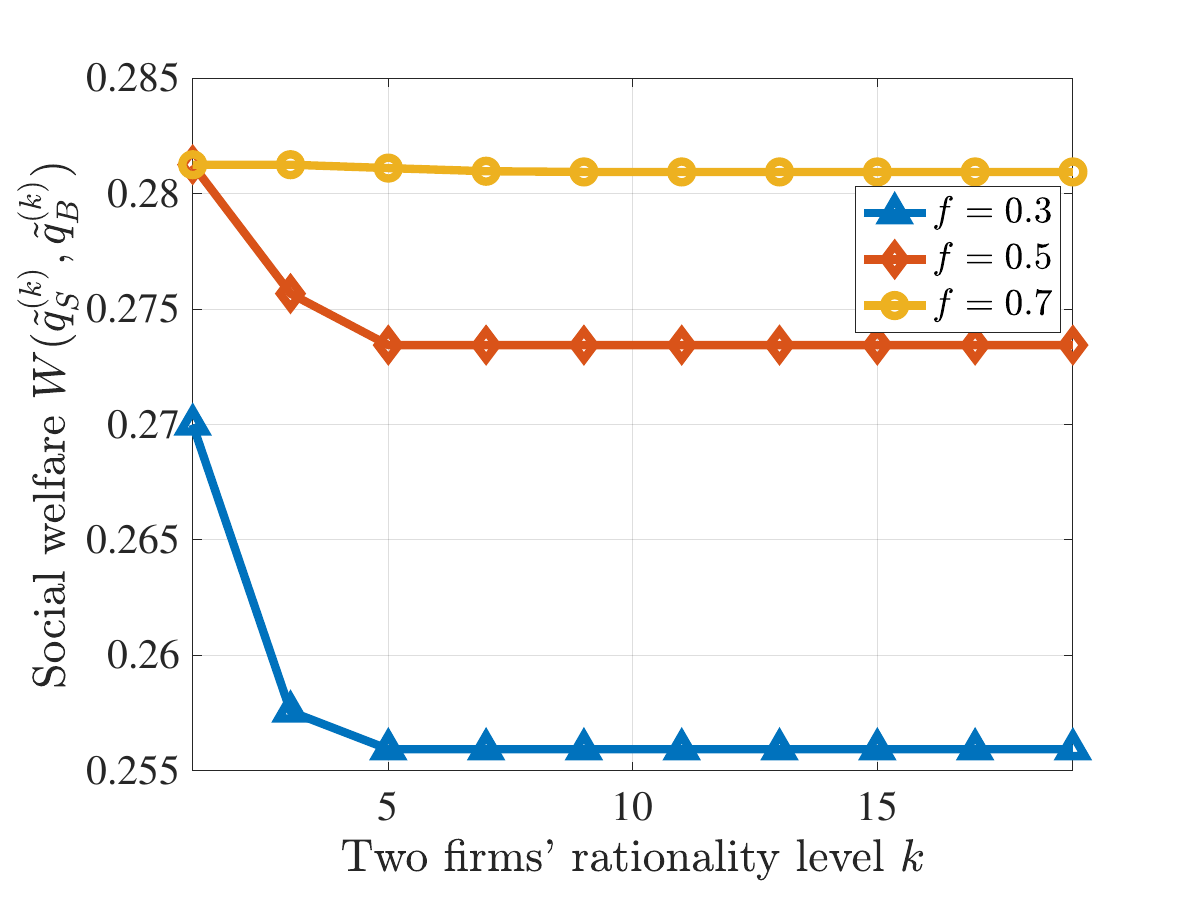}
        \label{fig: utility design-odd case}
    \end{minipage}
    }
    \caption{\ly{Price of rationality and welfare when both firms have level $k$ and the planner uses utility $U$ with the optimal weight $\gamma^*$ from Theorem \ref{thm: optimal cooperation level}.}}
\end{figure}

\ly{Finally, we compare Nash equilibrium and level-2 outcomes in two market structures: a standard duopoly, where both firms maximize profit, and a mixed duopoly, where firm $B$ maximizes welfare. We also include a mixed duopoly in which the planner uses $U$ with the optimal weight $\gamma^*$. Both firms have rationality level 2 in all level-2 cases. Figure \ref{fig: mixed duopoly comparison} summarizes the comparison. Figure \ref{fig: mixed duopoly social welfare} shows welfare as capacity varies. Comparing the blue and red curves, and the yellow and purple curves, shows that the planner's presence weakly improves welfare in both the Nash and level-2 settings. Comparing the blue and yellow curves, and the red and purple curves, shows that level-2 welfare can exceed the corresponding Nash benchmark at some capacities. The designed utility gives the highest welfare among the settings compared, with ties at some capacities. This illustrates how utility design can improve welfare beyond the gains from including a welfare-oriented firm alone.}

\ly{Figure \ref{fig: mixed duopoly profit} also shows that firm $S$ earns the lowest profit among the compared settings when the planner uses the designed utility $U$, with ties at some capacities. This occurs even at capacities where $\gamma^*>0$, as shown in Figure \ref{fig: gamma-even}. A positive weight on firm $S$'s profit therefore does not guarantee a higher realized profit for that firm. The planner selects the weight to maximize social welfare, and the weight changes both firms' outputs through the reasoning recursion. Cooperation in the utility function should thus be distinguished from an increase in firm $S$'s realized profit.}

\begin{figure}[ht]
\centering
    \subfigure[Comparison of social welfare]{
        \begin{minipage}[t]{0.49\linewidth}
        \centering
        \includegraphics[width=\linewidth]{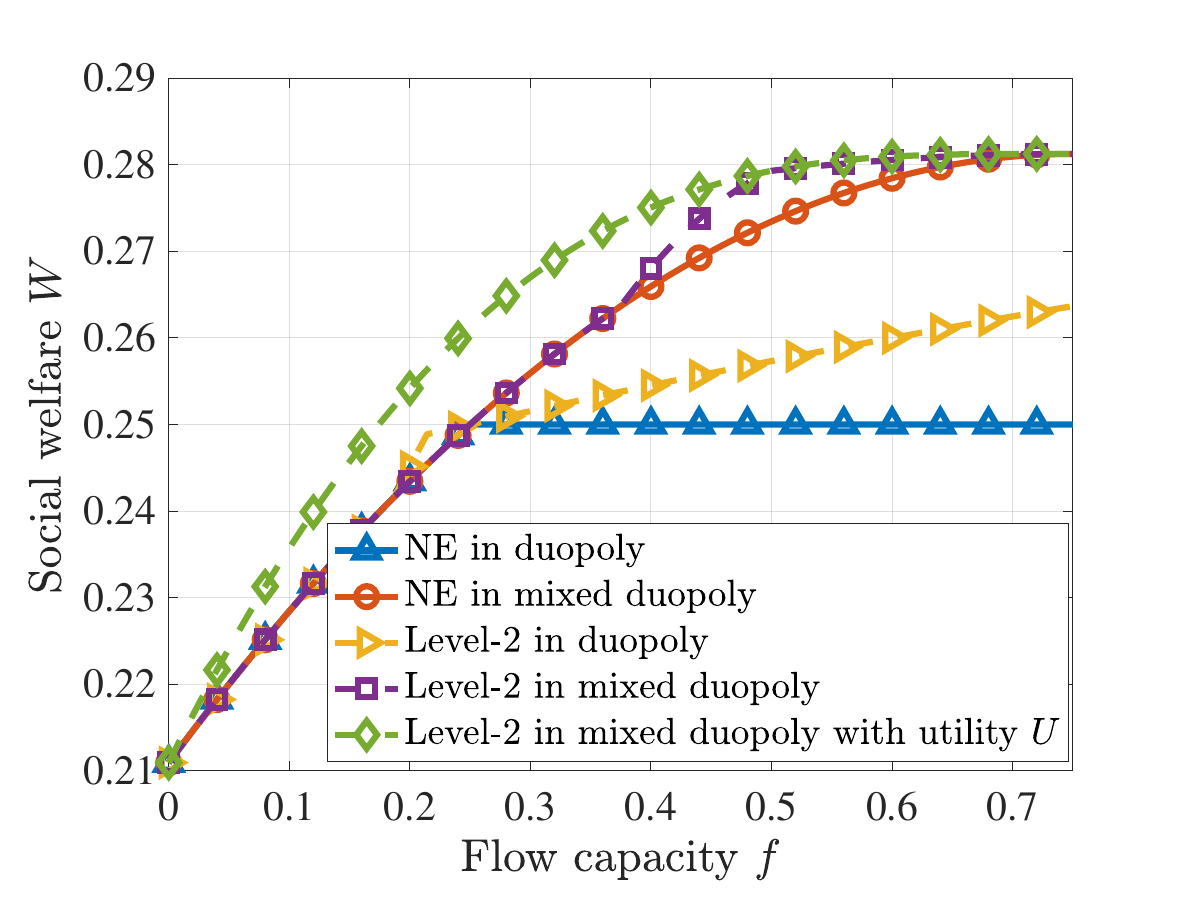}
        \label{fig: mixed duopoly social welfare}
    \end{minipage}
    }%
    \subfigure[Comparison of firm $S$'s profit]{
    \begin{minipage}[t]{0.49\linewidth}
        \centering
        \includegraphics[width=\linewidth]{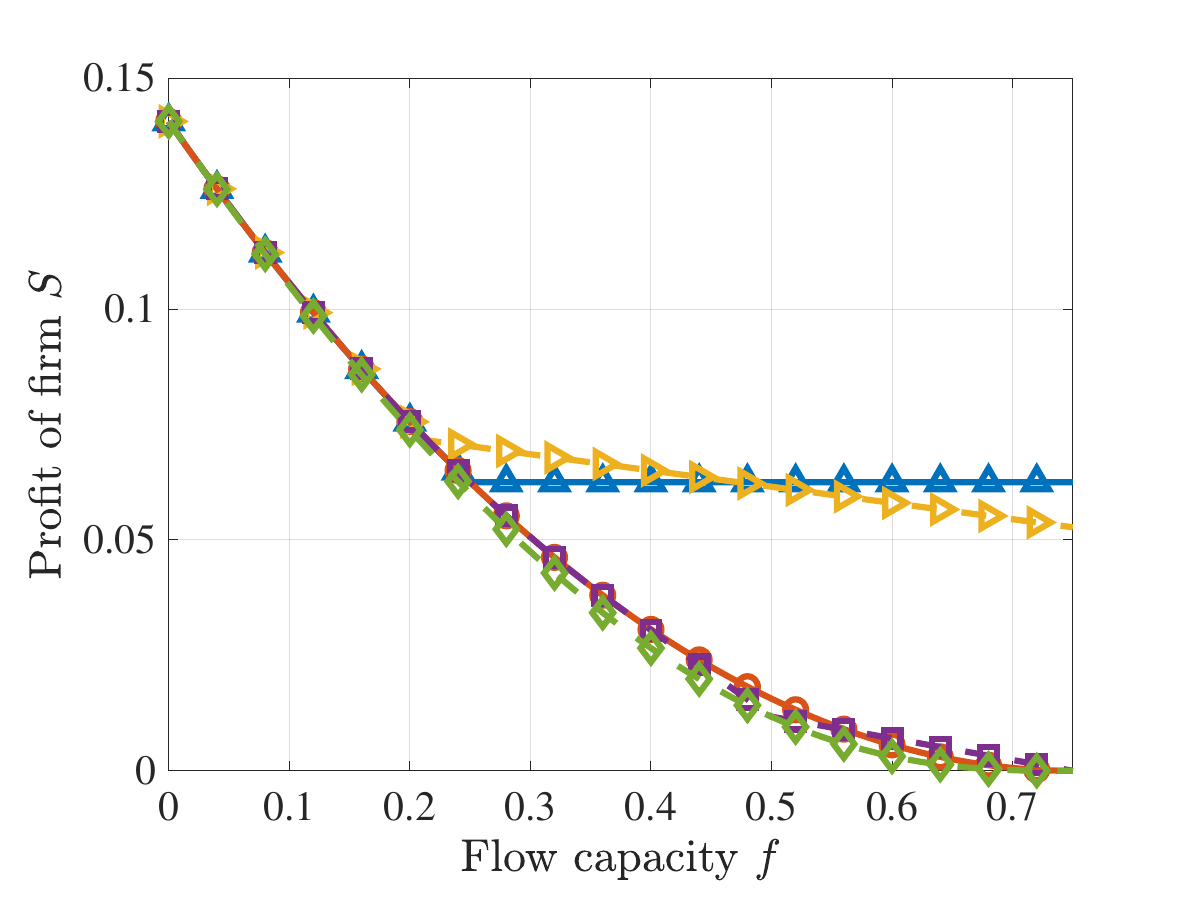}
        \label{fig: mixed duopoly profit}
    \end{minipage}
    }
    \caption{\ly{Comparison of welfare and firm $S$'s profit across five settings: Nash equilibrium in a standard duopoly, Nash equilibrium in a mixed duopoly, level-2 behavior in a standard duopoly, level-2 behavior in a mixed duopoly, and }\ly{level-2 behavior in a mixed duopoly with the designed utility $U$.}}
    \label{fig: mixed duopoly comparison}
\end{figure}

\section{Conclusion}\label{sec: conclusion}
\ly{We study how bounded rationality affects strategies and social welfare in mixed duopoly markets. The planner's strategy and the resulting welfare can differ substantially from the Nash equilibrium predictions. For a fixed level-$k$ opponent, level-$(k+1)$ behavior maximizes welfare, and welfare is weakly decreasing as the planner's level moves away from this optimum. We also derive planner strategies with complete, incomplete, or no distributional information about the opponent's level. Finally, we show that utility design can improve realized welfare by allowing cooperation with or opposition to the self-interested firm. Numerical studies illustrate the welfare differences, the value of information, and the role of the cooperation level. Future work can extend the analysis to more general production cost functions and examine how these costs affect market strategies and outcomes.}

% or


% \appendices

\ifCLASSOPTIONcaptionsoff
  \newpage
\fi



\bibliographystyle{IEEEtran}
\bibliography{refs}
\appendices

\section{Proof of Lemma \ref{lemma: closed-form level k behaviors}}\label{appendix: closed-form level k behaviors}

\begin{proof}
    We prove this by mathematical induction. First, we verify that the expressions are true when $k = 0$, according to the level-0 behaviors shown in Section \ref{sec: level-k reasoning}. Then, we assume $q_S^{(k)}$ and $q_B^{(k)}$ are true for any $k \in \mathbb{N}$ such that $k>0$ and $k$ is even.

    As stated in Section \ref{sec: level-k reasoning}, the level-$(k+1)$ firm $S$ will best respond to the level-$k$ firm $B$ and the level-$(k+1)$ firm $B$ will best respond to the level-$k$ firm $S$. Based on the best responses shown in Lemma \ref{lemma: best response functions}, the level-$(k+1)$ behavior of firm $S$ is written as follows,
    \begin{align*}
        q_S^{(k+1)} \!=\! \max\! \left\{\! \frac{b\!-\!c}{2a} \!-\! \min\!\left\{\! \frac{(2^{\frac{k}{2}}\!-\!1)(b-c)}{2^{\frac{k}{2}+1}a} \!+\! \frac{f}{2^{\frac{k}{2}+2}}, \frac{f}{2} \!\right\}, 0 \!\right\}.
    \end{align*}
    Equivalently, we have
    \begin{align*}
        q_S^{(k+1)} = \max \left\{ \max\left\{ \frac{b-c}{2^{\frac{k}{2}+1}a} - \frac{f}{2^{\frac{k}{2}+2}}, \frac{b-c}{2a}-\frac{f}{2} \right\}, 0 \right\}.
    \end{align*}
    Since $f \le (b-c)/a$, it can be further simplified as
    \begin{align*}
        q_S^{(k+1)} = \max\left\{ \frac{b-c}{2^{\frac{k}{2}+1}a} - \frac{f}{2^{\frac{k}{2}+2}}, \frac{b-c}{2a}-\frac{f}{2} \right\},
    \end{align*}
    which is consistent with the closed-form solution of $q_S^{(k+1)}$. Similarly, we may derive the level-$(k+1)$ behavior of firm $B$ by $q_B^{(k+1)} = \mathcal{B}_B(q_S^{(k)})$ and obtain
    \begin{align*}
        q_B^{(k+1)} \!=\! \max\! \left\{\! \min\!\left\{\! \frac{b\!-\!c}{a}\!-\!\max\!\left\{ \frac{b-c}{2^{\frac{k}{2}+1}a}, \frac{b\!-\!c}{2a}\!-\!\frac{f}{2} \!\right\}, f \!\right\}, 0 \!\right\}\!,
    \end{align*}
    Equivalently, we have
    \begin{align*}
        q_B^{(k+1)} \!=\! \max\! \left\{\! \min\left\{\! \frac{(2^{\frac{k}{2}+1}-1)(b-c)}{2^{\frac{k}{2}+1}a}, \frac{b-c}{2a}+\frac{f}{2}, f \!\right\}, 0 \!\right\}.
    \end{align*}
    Since $f \le (b-c)/a$, the preceding equation becomes
    \begin{align*}
        q_B^{(k+1)} = \min\left\{ \frac{(2^{\frac{k}{2}+1}-1)(b-c)}{2^{\frac{k}{2}+1}a}, f \right\},
    \end{align*}
    which is consistent with the closed-form solution of $q_B^{(k+1)}$. The case where $k$ is odd can be established in a similar way. Therefore, we see that the level-$(k+1)$ expressions are true. By mathematical induction, the proof is complete.
\end{proof}

\section{Proof of Theorem \ref{thm: level-k NE comparison}}\label{appendix: supplement}

For ease of notation, we let $f = \beta(b-c)/a$ with $\beta \in [0,1]$ since $0 \le f \le (b-c)/a$. We first assume $k>0$ is even and $\Delta=0$, then
    \begin{align}\label{eq:qsqb}
        q_S^{(k)} \!+\! q_B^{(k)} = \max\Big\{\! \min\left\{ x, y \right\}, \min\left\{ x-y+z, z \right\} \!\Big\},
    \end{align}
    where
    \begin{align*}
        \begin{dcases}
            x = \left( 1-\frac{1}{2^{\frac{k}{2}+1}}+\frac{\beta}{2^{\frac{k}{2}+1}} \right)\frac{b-c}{a},\\
            y = \left( \frac{1}{2^{\frac{k}{2}+1}}+\beta \right)\frac{b-c}{a},\\
            z = \left( \frac{1}{2}+\frac{\beta}{2} \right)\frac{b-c}{a}.
        \end{dcases}
    \end{align*}
    Since $\beta \in [0,1]$, we see that {$x \le (b-c)/a$} and {$z \le (b-c)/a$}. {Thus, Eq.\ \eqref{eq:qsqb} implies $q_S^{(k)} + q_B^{(k)} \le \max\{x,z\} \le {(b-c)/a}$ and therefore}
    $$d(q_S^{(k)}, q_B^{(k)}) = {(b-c)/a}-(q_S^{(k)} + q_B^{(k)}).$$
    Recall the total produced quantity at Nash equilibrium as $q_S^{NE} + q_B^{NE} = z$, we see that $d(q_S^{NE},q_B^{NE}) = {(b-c)/a}-(q_S^{NE}+q_B^{NE})$. By Lemma \ref{lemma: equivalent social welfare measure}, the level-$k$ performance is better than the equilibrium performance if and only if
    $$q_S^{(k)} + q_B^{(k)} > q_S^{NE}+q_B^{NE}.$$
   We combine Eq. \ref{eq:qsqb}  with the observation $\min\left\{ x-y+z, z \right\} \le z$, the previous inequality is satisfied if and only if
    $$\min\left\{ x, y \right\} > z,$$
    which implies that $(1-\frac{1}{2^{\frac{k}{2}}}) < \beta < 1$ since $k>0$. When $k\in\mathbb{N}$ is even and $\Delta=1$, similar techniques can be established and the same condition is obtained. For $\Delta \ge 2$, we first write the production levels of two suppliers as follows,
    \begin{align*}
        \begin{dcases}
            q_S^{(k)} = \max\left\{ \frac{1}{2^{\frac{k}{2}+1}}\frac{b-c}{a}, \left(\frac{1}{2}-\frac{\beta}{2}\right)\frac{b-c}{a} \right\},\\
            q_B^{(k+\Delta)} = \min\left\{ \left(1-\frac{1}{2^{\frac{k+\Delta}{2}}}+\frac{\beta}{2^{\frac{k+\Delta}{2}+1}}\right)\frac{b-c}{a}, \beta\frac{b-c}{a} \right\}.
        \end{dcases}
    \end{align*}
    By adding them up and embedding minimum into maximum, we obtain the total production level as
    \begin{align}
        q_S^{(k)} \!+\! q_B^{(k+\Delta)} = \max\Big\{\! \min\left\{ x, y \right\}, \min\left\{ x-y+z, z \right\} \!\Big\},
    \end{align}
    where
    \begin{align*}
        \begin{dcases}
            x = \left( 1-\frac{1}{2^{\frac{k+\Delta}{2}}}+\frac{1}{2^{\frac{k}{2}+1}}+\frac{\beta}{2^{\frac{k+\Delta}{2}+1}} \right)\frac{b-c}{a},\\
            y = \left( \frac{1}{2^{\frac{k}{2}+1}}+\beta \right)\frac{b-c}{a},\\
            z = \left( \frac{1}{2}+\frac{\beta}{2} \right)\frac{b-c}{a}.
        \end{dcases}
    \end{align*}
    Denote
    \begin{align*}
        \begin{dcases}
            T_{xy} = 1-\frac{1}{2^{\frac{k+\Delta}{2}+1}-1},\\
            T_{zy} = 1-\frac{1}{2^{\frac{k}{2}}}.
        \end{dcases}
    \end{align*}
    It can be verified that $x \le y$ if and only if $\beta \ge T_{xy}$, and $y \le z$ if and only if $\beta \le T_{zy}$. Furthermore, as we assume $\Delta \ge 2$, the inequality $2^{\frac{k}{2}}+1 \le 2^{\frac{k+\Delta}{2}+1}$ holds. It then implies
    $$\frac{1}{2^{\frac{k+\Delta}{2}+1}-1} \le \frac{1}{2^{\frac{k}{2}}}.$$
    Thus, we have $T_{zy} \le T_{xy}$. We partition $\beta \in [0,1]$ into three intervals, the total production level $q_S^{(k)} + q_B^{(k+\Delta)}$ can be seen as a piece-wise linear function of $\beta$.
    \begin{enumerate}[label=(\roman*)]
        \item if $\beta \le T_{zy}$: we have $x \ge y$, $y \le z$ and thus
        $$q_S^{(k)} + q_B^{(k+\Delta)} = \max\{y,z\} = z;$$
        \item if $T_{zy} \le \beta \le T_{xy}$: we have $x \ge y$, $y \ge z$ and thus
        $$q_S^{(k)} + q_B^{(k+\Delta)} = \max\{y,z\} = y;$$
        \item if $\beta \ge T_{xy}$: we have $x \le y$, $y \ge z$ and thus
        $$q_S^{(k)} + q_B^{(k+\Delta)} = \max\{x,x-y+z\} = x.$$
    \end{enumerate}
    Since $q_S^{NE} + q_B^{NE} = z$, we observe that the total production level $q_S^{(k)} + q_B^{(k+\Delta)}$ is the same as equilibrium in interval (i), increases with a steeper slope in interval (ii), and then keeps increasing in interval (iii). The function is illustrated in Fig. \ref{fig: supplement}.

    \begin{figure}[ht]
    \centering
    \includegraphics[width=\linewidth]{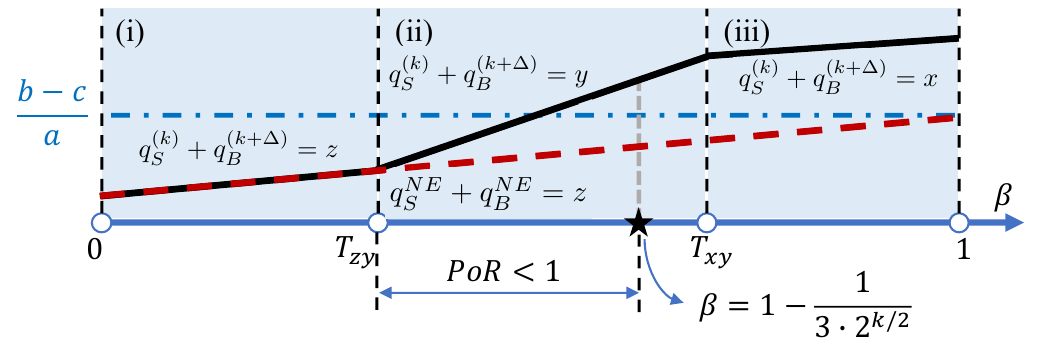}
    \caption{Illustration of the level-$k$ total production level, as a piece-wise linear function of $\beta$.}
    \label{fig: supplement}
    \end{figure}

    According to Lemma \ref{lemma: equivalent social welfare measure}, to make the level-$k$ performance strictly better than the equilibrium performance, we need the distance between $q_S^{(k)} + q_B^{(k+\Delta)}$ and $(b-c)/a$ to be strictly less than the distance between $q_S^{NE} + q_B^{NE}$ and $(b-c)/a$. From Fig. \ref{fig: supplement}, we observe that such a region of $\beta$ exists if $\beta$ is slightly larger than $T_{zy}$. Next, we assume the boundary condition occurs in interval (ii) and thus let $y-(b-c)/a = (b-c)/a-(q_S^{NE} + q_B^{NE})$, yield
    $$\beta = 1-\frac{1}{3\cdot2^{\frac{k}{2}}} \le 1-\frac{1}{2^{\frac{k}{2}+2}-1} \le 1-\frac{1}{2^{\frac{k+\Delta}{2}+1}-1} = T_{xy},$$
    since $\Delta \ge 2$. Thus, the boundary condition indeed occurs in interval (ii).

    When $T_{zy} < \beta < 1-\frac{1}{3\cdot2^{\frac{k}{2}}}$, it can be seen that $q_S^{(k)} + q_B^{(k+\Delta)}$ is closer to $(b-c)/a$ compared to $q_S^{NE} + q_B^{NE}$, implying that level-$k$ performance is strictly better than equilibrium performance. On the other hand, when $\beta > 1-\frac{1}{3\cdot2^{\frac{k}{2}}}$, $q_S^{(k)} + q_B^{(k+\Delta)}$ is monotonically deviating from $(b-c)/a$ while $q_S^{NE} + q_B^{NE}$ is monotonically approaching $(b-c)/a$. Thus, in this case, the equilibrium performance is better than level-$k$ performance.

    Therefore, the sufficient and necessary condition for the case of $k$ being even and $\Delta \ge 2$ is $\beta \in (1-\frac{1}{2^{\frac{k}{2}}}, 1-\frac{1}{3\cdot2^{\frac{k}{2}}})$. When $k$ is an odd number, the statement can be established similarly.

\section{Supporting Lemmas for Utility Design}\label{appendix: utility design lemmas}

\begin{lemma}\label{lemma: best response function utility design}
    The best response of the benevolent social planner with the utility function $U(q_S,q_B)$ can be expressed compactly as follows,
    \begin{align*}
        \tilde{\mathcal{B}}_B(q_S) = \max \Big\{ \min\big\{ (b-c)/a-(1+\gamma)q_S, f \big\}, 0 \Big\}.
    \end{align*}
\end{lemma}
\begin{proof}
    We start by rewriting the utility function $U(q_S, q_B)$,
    \begin{align*}
        U(q_S, q_B) = &-\frac{a}{2}q_B^2 + \big((b-c)-a(1+\gamma)q_S\big)q_B\\
        &+ aq_S\left(\frac{(1+\gamma)(b-c)}{a}-(\frac{1}{2}+\gamma)q_S\right) + m.
    \end{align*}
    The best response of the benevolent social planner is solved by
    $\tilde{\mathcal{B}}_B(q_S) \in \arg\max_{0 \le q_B \le f} \;U(q_S,q_B)$. Since $U(q_S, q_B)$ is concave in $q_B$, we combine the first-order condition with constraints on produced quantity and then the closed-form expression is derived.
\end{proof}

\begin{lemma}\label{lemma: closed-form level k behaviors utility design, gamma<-1}
    With the utility function $U$ and $\gamma<-1$, for any $k \in \mathbb{N}$, the level-$k$ behavior of firm $S$ is given by
    \begin{align*}
        \tilde{q}_S^{(k)} = \begin{dcases}
            (b-c)/2a,&\text{if $k=0$}; \\
            (b-c)/2a-f/4,&\text{if $k=1$}; \\
            (b-c)/2a-f/2,&\text{if $k \ge 2$}.
        \end{dcases}
    \end{align*}
    Similarly, the level-$k$ behavior of firm $B$ is given by
    \begin{align*}
        \tilde{q}_B^{(k)} = \begin{dcases}
            f/2,&\text{if $k=0$}; \\
            f,&\text{if $k \ge 1$}.
        \end{dcases}
    \end{align*}
\end{lemma}
\begin{proof}
    Since $\gamma<-1$ and $f \le (b-c)/a$, the best response of the social planner in Lemma \ref{lemma: best response function utility design} can be simplified as $\tilde{\mathcal{B}}_B(q_S) = f$. Then, the closed-form solutions follow naturally.
\end{proof}

\begin{lemma}\label{lemma: closed-form level k behaviors utility design, -1<gamma<1}
    With the utility function $U$ and $-1 \le \gamma \le 1$, for any $k \in \mathbb{N}$ that is even, the level-$k$ behaviors are given by\footnote{{Note that we consider $0^0=1$.}}
    \begin{align*}
        \begin{dcases}
            &\tilde{q}_S^{(k)} = \max\left\{ \frac{(1+\gamma)^{\frac{k}{2}}}{2^{\frac{k}{2}+1}} \frac{b-c}{a}, \frac{b-c}{2a}-\frac{f}{2} \right\},\\
            &\tilde{q}_B^{(k)} = \min\left\{ \left( 1-\frac{(1+\gamma)^{\frac{k}{2}}}{2^{\frac{k}{2}}} \right) \frac{b-c}{a} + \frac{(1+\gamma)^{\frac{k}{2}}}{2^{\frac{k}{2}+1}}f, f \right\}.
        \end{dcases}
    \end{align*}
    For any $k \in \mathbb{N}$ that is odd, the level-$k$ behaviors are given by
    \begin{align*}
        \begin{dcases}
            &\tilde{q}_S^{(k)} = \max\left\{ \frac{(1+\gamma)^{\frac{k-1}{2}}}{2^{\frac{k+1}{2}}} \left( \frac{b-c}{a}-\frac{f}{2} \right), \frac{b-c}{2a}-\frac{f}{2} \right\},\\
            &\tilde{q}_B^{(k)} = \min\left\{ \left( 1-\frac{(1+\gamma)^{\frac{k+1}{2}}}{2^{\frac{k+1}{2}}} \right) \frac{b-c}{a}, f \right\}.
        \end{dcases}
    \end{align*}
\end{lemma}
\begin{proof}
    We prove this by mathematical induction. First, we verify that the closed-form solutions satisfy $\tilde{q}_S^{(0)} = (b-c)/2a$ and $\tilde{q}_B^{(0)} = f/2$ by letting $k = 0$. Then, we assume the expressions of level-$k$ behavior are true for any $k \in \mathbb{N}$ such that $k>0$ and $k$ is an even number. The level-$(k+1)$ firm $S$ will best respond to the level-$k$ firm $B$, based on the best response function $\mathcal{B}_S(q_B)$ as shown in Lemma \ref{lemma: best response functions}, i.e.,
    \begin{align*}
        \tilde{q}_S^{(k+1)} = \max \left\{ (b-c-a\tilde{q}_B^{(k)})/2a, 0 \right\},
    \end{align*}
    which is equivalent to
    \begin{align*}
        \tilde{q}_S^{(k+1)} \!=\! \max \!\left\{\! \max\!\left\{\! \frac{(1+\gamma)^{\frac{k}{2}}}{2^{\frac{k}{2}+1}} \!\left(\! \frac{b-c}{a}\!-\!\frac{f}{2} \!\right)\!, \frac{b-c}{2a}\!-\!\frac{f}{2} \!\right\}, 0 \!\right\}.
    \end{align*}
    Since $f \le (b-c)/a$, it follows that
    \begin{align*}
        \tilde{q}_S^{(k+1)} = \max \left\{ \frac{(1+\gamma)^{\frac{k}{2}}}{2^{\frac{k}{2}+1}} \left( \frac{b-c}{a}-\frac{f}{2} \right), \frac{b-c}{2a}-\frac{f}{2}\right\},
    \end{align*}
    which is consistent with the closed-form expression of $\tilde{q}_S^{(k+1)}$. Similarly, the level-$(k+1)$ benevolent social planner will best respond to the level-$k$ self-interested firm $S$, based on the best response $\tilde{\mathcal{B}}_B(q_S)$ shown in Lemma \ref{lemma: best response function utility design}, i.e.,
    \begin{align*}
        \tilde{q}_B^{(k+1)} = \max \left\{ \min\left\{ (b-c)/a-(1+\gamma)\tilde{q}_S^{(k)}, f \right\}, 0 \right\}.
    \end{align*}
    By plugging in $\tilde{q}_S^{(k)}$, the preceding equation is expanded as
    \begin{align*}
        \tilde{q}_B^{(k+1)} \!=\! \max\! \left\{\! \min\!\left\{\!\! \left(\! 1 \!-\! \frac{(1\!+\!\gamma)^{\frac{k}{2}+1}}{2^{\frac{k}{2}+1}} \!\right) \!\frac{b\!-\!c}{a}, \frac{b\!-\!c}{2a} \!+\! \frac{f}{2}\!,
        f \!\right\}\!, 0 \!\right\}\!,
    \end{align*}
    since $\gamma \ge -1$. As we assume $f \le (b-c)/a$, it is equivalent to
    \begin{align*}
        \tilde{q}_B^{(k+1)} = \max \left\{ \min\left\{ \left( 1-\frac{(1+\gamma)^{\frac{k}{2}+1}}{2^{\frac{k}{2}+1}} \right) \frac{b-c}{a}, f \right\}, 0 \right\},
    \end{align*}
    Furthermore, since $\gamma \le 1$, we have
    \begin{align*}
        \tilde{q}_B^{(k+1)} = \min\left\{ \left( 1-\frac{(1+\gamma)^{\frac{k}{2}+1}}{2^{\frac{k}{2}+1}} \right) \frac{b-c}{a}, f \right\}.
    \end{align*}
    This is also consistent with the closed-form solution of $\tilde{q}_B^{(k+1)}$. The case of $k$ being odd can be established by similar techniques. Therefore, we see that the level-$(k+1)$ expressions are true. The proof is then complete.
\end{proof}

\begin{lemma}\label{lemma: closed-form level k behaviors utility design, gamma>1}
    With the utility function $U$ and $\gamma>1$, for any $k \in \mathbb{N}$ that is even, the level-$k$ behaviors are given by
    \begin{align*}
    \begin{dcases}
        &\tilde{q}_S^{(k)} = \frac{b-c}{2a},\\
        &\tilde{q}_B^{(k)} = \max\left\{ \left( 1-\frac{(1+\gamma)^{\frac{k}{2}}}{2^{\frac{k}{2}}} \right) \frac{b-c}{a} + \frac{(1+\gamma)^{\frac{k}{2}}}{2^{\frac{k}{2}+1}}f, 0 \right\}.
    \end{dcases}
    \end{align*}
    For any $k \in \mathbb{N}$ that is odd, the level-$k$ behaviors are given by
    \begin{align*}
    \begin{dcases}
        &\tilde{q}_S^{(k)} = \min\left\{ \frac{(1+\gamma)^{\frac{k-1}{2}}}{2^{\frac{k+1}{2}}} \left( \frac{b-c}{a}-\frac{f}{2} \right), \frac{b-c}{2a} \right\},\\
        &\tilde{q}_B^{(k)} = 0.
    \end{dcases}
    \end{align*}
\end{lemma}
\begin{proof}
    It can be established similarly to Lemma \ref{lemma: closed-form level k behaviors utility design, -1<gamma<1}.
\end{proof}

\section{Proof of Theorem \ref{thm: optimal cooperation level}(i)}\label{appendix: optimal cooperation level even}

\begin{proof}
    We first assume $k \in \mathbb{N}^+$ is an even number and denote $\varphi = ( (1+\gamma)/2)^{k/2}$. If $\gamma < -1$, by Lemma \ref{lemma: closed-form level k behaviors utility design, gamma<-1} and $f \le (b-c)/a$, there holds
    \begin{align*}
        \tilde{q}_S^{(k)} + \tilde{q}_B^{(k)} = (b-c)/2a + f/2 \le (b-c)/a.
    \end{align*}
    If $\gamma > 1$, by Lemma \ref{lemma: closed-form level k behaviors utility design, gamma>1},
    \begin{align*}
        \tilde{q}_S^{(k)} + \tilde{q}_B^{(k)} &= \max\left\{ \frac{3(b-c)}{2a} - \varphi \left( \frac{b-c}{a}-\frac{f}{2} \right), \frac{b-c}{2a} \right\} \\
        &\le \frac{b-c}{2a}+\frac{f}{2} \le \frac{b-c}{a}.
    \end{align*}
    If $-1 \le \gamma \le 1$, by Lemma \ref{lemma: closed-form level k behaviors utility design, -1<gamma<1}, we have
    \begin{align}\label{eq: temp2}
        \tilde{q}_S^{(k)} + \tilde{q}_B^{(k)} = \max\Big\{ \min\left\{ x, y \right\}, \min\left\{ x-y+z, z \right\} \Big\},
    \end{align}
    where $x = (1-\frac{\varphi}{2})\frac{b-c}{a} + \frac{\varphi}{2}f$, $y = \frac{\varphi}{2}\frac{b-c}{a} + f$ and $z = \frac{b-c}{2a}+\frac{f}{2}$. Since $f \le (b-c)/a$, we can see that $x \le (1-\frac{\varphi}{2})\frac{b-c}{a} + \frac{\varphi}{2}\frac{b-c}{a} = \frac{b-c}{a}$. Thus, there holds $\tilde{q}_S^{(k)} + \tilde{q}_B^{(k)} \le (b-c)/a$ for any $\gamma \in \mathbb{R}$ and thus
    $$d(\tilde{q}_S^{(k)}, \tilde{q}_B^{(k)}) = (b-c)/a-(\tilde{q}_S^{(k)} + \tilde{q}_B^{(k)}).$$
    According to Lemma \ref{lemma: equivalent social welfare measure}, the maximization problem \eqref{eq: max problem utility design} is equivalent to
    \begin{align*}
        \gamma^* \in \underset{\gamma \in \mathbb{R}}{\arg\max} \;\;\tilde{q}_S^{(k)} + \tilde{q}_B^{(k)}, \quad\forall k \in \mathbb{N}.
    \end{align*}
    For any $\gamma < -1$ or $\gamma > 1$, we have $$\tilde{q}_S^{(k)} + \tilde{q}_B^{(k)} \le (b-c)/2a+f/2 = z.$$
    For $\gamma \in [-1,1]$, we discuss two cases separately. First, consider the case where $x \le y$, or equivalently, $\frac{(b-c)/a-f}{(b-c)/a-f/2} \le \varphi \le 1$. Eq. \eqref{eq: temp2} can be simplified as
    \begin{align*}
        \tilde{q}_S^{(k)} + \tilde{q}_B^{(k)} \!=\! \max\{x, x-y+z\} \!=\! \frac{b-c}{a} - \frac{1}{2}\left(\!\frac{b-c}{a}-f\!\right)\varphi,
    \end{align*}
    which is linear in $\varphi$. Therefore, in this case, the maximizer of $\tilde{q}_S^{(k)} + \tilde{q}_B^{(k)}$ satisfies
    \begin{align}\label{eq: global maximizer}
        \varphi^* = \left(\frac{1+\gamma^*}{2}\right)^{\frac{k}{2}} = \frac{(b-c)/a-f}{(b-c)/a-f/2},
    \end{align}
    with the maximal objective value as
    \begin{align}\label{eq: global maximal objective}
        (\tilde{q}_S^{(k)} + \tilde{q}_B^{(k)})^* = \frac{((b-c)/a)^2 + f(b-c)/a - f^2}{2(b-c)/a - f}.
    \end{align}
    Second, consider the case where $x \ge y$, or equivalently, $0 \le \varphi \le \frac{(b-c)/a-f}{(b-c)/a-f/2}$. Eq. \eqref{eq: temp2} can be simplified as
    \begin{align*}
        \tilde{q}_S^{(k)} + \tilde{q}_B^{(k)} = \max\{y, z\} = \max\left\{ f+\frac{\varphi}{2}\frac{b-c}{a}, \frac{b-c}{2a}+\frac{f}{2} \right\},
    \end{align*}
    which is linear in $\varphi$. Then, we obtain the same maximizer and maximal objective value as in the first case. Furthermore, we observe $(\tilde{q}_S^{(k)} + \tilde{q}_B^{(k)})^* \ge z$. Hence, the global maximizer and the maximal objective value are given by Eq. \eqref{eq: global maximizer} and Eq. \eqref{eq: global maximal objective}, respectively.

    When $k$ is an odd number, the proof has more cases, and we leave it in Appendix \ref{appendix: optimal cooperation level}.
\end{proof}

\section{Proof of Theorem \ref{thm: optimal cooperation level}(ii)}\label{appendix: optimal cooperation level}

Recall that we aim to solve the following problem
\begin{align*}
    \underset{\gamma \in \mathbb{R}}{\max} \;W(\tilde{q}_S^{(k)},\tilde{q}_B^{(k)}),
\end{align*}
for any odd number $k \in \mathbb{N}$. Based on Lemma \ref{lemma: equivalent social welfare measure}, it is also equivalent to considering
\begin{align*}
    \min_{\gamma \in \mathbb{R}} \; d(\tilde{q}_S^{(k)},\tilde{q}_B^{(k)}) = \big| (b-c)/a - (\tilde{q}_S^{(k)}+\tilde{q}_B^{(k)}) \big|.
\end{align*}
This indicates that we aim to find $\gamma$ such that the total production level $(\tilde{q}_S^{(k)}+\tilde{q}_B^{(k)})$ is as close to the optimal social total production $(b-c)/a$ as possible. For ease of notation, we will abbreviate $d(\tilde{q}_S^{(k)}, \tilde{q}_B^{(k)})$ to $d$ and let $f=\beta(b-c)/a$ with $\beta \in [0,1]$ since $0 \le f \le (b-c)/a$. Denote
\begin{align*}
    &x = \left[ 1 - \frac{\gamma(1+\gamma)^{\frac{k-1}{2}}}{2^{\frac{k+1}{2}}} - \frac{(1+\gamma)^{\frac{k-1}{2}}}{2^{\frac{k+1}{2}}}\frac{\beta}{2}\right] \frac{b-c}{a},\\
    &y = \left[ \beta+\frac{(1+\gamma)^{\frac{k-1}{2}}}{2^{\frac{k+1}{2}}} \left( 1-\frac{\beta}{2} \right)\right] \frac{b-c}{a},\\
    &z = \left( \frac{1}{2}+\frac{\beta}{2} \right)\frac{b-c}{a}.
\end{align*}
We first show the case of $k=1$, and then prove the statement for a general $k$ with $k>1$.

\begin{lemma}\label{lemma: A0}
    When $k=1$, an optimal cooperation level is given by
    \begin{align*}
        \gamma^* = \max\left\{ -\frac{\beta}{2}, 2\left( 1-\beta \right)-1 \right\}.
    \end{align*}
\end{lemma}
\begin{proof}
    We first assume $\gamma < -1$, apply $k=1$ to Lemma \ref{lemma: closed-form level k behaviors utility design, gamma<-1} and obtain $\tilde{q}_S^{(1)}+\tilde{q}_B^{(1)} = \left( \frac{1}{2}+\frac{3\beta}{4} \right)\frac{b-c}{a}$. Second, if $\gamma > 1$, by Lemma \ref{lemma: closed-form level k behaviors utility design, gamma>1}, we have $\tilde{q}_S^{(1)}+\tilde{q}_B^{(1)} = \left( \frac{1}{2}-\frac{\beta}{4} \right)\frac{b-c}{a}$, which is farther from the optimal social total production since $\beta \in [0,1]$.

    Next, if $-1 \le \gamma \le 1$, there holds
    \begin{align*}
        \tilde{q}_S^{(1)} \!+\! \tilde{q}_B^{(1)} = \min\left\{\! \left( 1-\frac{\gamma}{2}-\frac{\beta}{4} \right)\frac{b-c}{a}, \left( \frac{1}{2}+\frac{3\beta}{4} \right)\frac{b-c}{a} \!\right\}.
    \end{align*}
    Depending on which term in the preceding expression is greater, we discuss the following two cases.

    \textit{i) If $-1 \le \gamma \le 2(1-\beta)-1$:} we have
    \begin{align*}
        \tilde{q}_S^{(1)}+\tilde{q}_B^{(1)} = \left( \frac{1}{2}+\frac{3\beta}{4} \right)\frac{b-c}{a},
    \end{align*}
    regardless of $\gamma$.

    \textit{ii) If $2(1-\beta)-1 \le \gamma \le 1$:} we have
    \begin{align*}
        \tilde{q}_S^{(1)}+\tilde{q}_B^{(1)} = \left( 1-\frac{\gamma}{2}-\frac{\beta}{4} \right)\frac{b-c}{a}.
    \end{align*}
    Thus, the problem we aim to solve can be rewritten as
    \begin{align*}
        \min_{2(1-\beta)-1 \le \gamma \le 1} \; d^2(\tilde{q}_S^{(k)},\tilde{q}_B^{(k)}) = \left( \frac{\gamma}{2}+\frac{\beta}{4} \right)^2 \left(\frac{b-c}{a}\right)^2.
    \end{align*}
    Since the objective function is convex, the minimum point is $\gamma^* = -\beta/2$. The corresponding total production level is then equal to $(b-c)/a$. In the case where the minimal point is not feasible, we have $\gamma^* = 2(1-\beta)-1$ instead, since $\beta \in [0,1]$. In this case, the total production level is $\tilde{q}_S^{(1)}+\tilde{q}_B^{(1)} = \left( \frac{1}{2}+\frac{3\beta}{4} \right)\frac{b-c}{a}$, which is no worse than that achieved by assuming $\gamma<-1$ or $\gamma>1$. Therefore, an optimal cooperation level is given by
    \begin{align*}
        \gamma^* = \max\left\{ -\frac{\beta}{2}, 2\left( 1-\beta \right)-1 \right\}.
    \end{align*}
\end{proof}

Therefore, Theorem \ref{thm: optimal cooperation level}(ii) is established for the case of $k=1$. In the remainder of the proof, we assume $k>1$. The following lemma shows a compact representation of the total production level.

\begin{lemma}\label{lemma: A1}
    For $-1 \le \gamma \le 1$, we have the total production level as
    \begin{align*}
        \tilde{q}_S^{(k)}+\tilde{q}_B^{(k)} = \max\big\{ \min\{x,y\}, \min\{x-y+z, z\} \big\}.
    \end{align*}
\end{lemma}
\begin{proof}
    Based on Lemma \ref{lemma: closed-form level k behaviors utility design, -1<gamma<1}, we have
    \begin{align*}
        \begin{dcases}
            &\!\!\!\!\!\! \tilde{q}_S^{(k)} \!=\! \max\left\{\! \frac{(1+\gamma)^{\frac{k-1}{2}}}{2^{\frac{k+1}{2}}} \!\left(\! 1-\frac{\beta}{2} \!\right)\frac{b-c}{a}, \left(\!\frac{1}{2}-\frac{\beta}{2}\!\right)\frac{b-c}{a} \!\right\},\\
            &\tilde{q}_B^{(k)} = \min\left\{ \left( 1-\frac{(1+\gamma)^{\frac{k+1}{2}}}{2^{\frac{k+1}{2}}} \right) \frac{b-c}{a}, \beta\frac{b-c}{a} \right\}.
        \end{dcases}
    \end{align*}
    Summing up two production levels, we obtain
    \begin{align*}
        \tilde{q}_S^{(k)}+\tilde{q}_B^{(k)} = \max\left\{ \frac{(1+\gamma)^{\frac{k-1}{2}}}{2^{\frac{k+1}{2}}} \left( 1-\frac{\beta}{2} \right)\frac{b-c}{a} + \tilde{q}_B^{(k)}, \right.\\
        \left.\left(\frac{1}{2}-\frac{\beta}{2}\right)\frac{b-c}{a} + \tilde{q}_B^{(k)}\right\}.
    \end{align*}
    Through calculations, it is further simplified as
    \begin{align*}
        \tilde{q}_S^{(k)}+\tilde{q}_B^{(k)} = \max\big\{ \min\{x,y\}, \min\{x-y+z, z\} \big\}.
    \end{align*}
\end{proof}

To further simplify the expression of the total production level, we derive the sufficient and necessary conditions such that $x \le y$ and $y \le z$, respectively. For simplicity, we use the following shorthand,
\begin{align*}
    &T_{xy} = 2\left( 1-\beta \right)^{\frac{2}{k+1}}-1,\\
    &T_{zy} = 2\left( \frac{1-\beta}{1-\beta/2} \right)^{\frac{2}{k-1}}-1.
\end{align*}

\begin{lemma}\label{lemma: A2}
    For $-1 \le \gamma \le 1$, $x \le y$ if and only if $\gamma \ge T_{xy}$, while $y \le z$ if and only if $\gamma \le T_{zy}$.
\end{lemma}

\begin{proof}
    By definitions, $x \le y$ if and only if
    \begin{align*}
        \beta \ge 1-\frac{(1+\gamma)^{\frac{k+1}{2}}}{2^{\frac{k+1}{2}}}.
    \end{align*}
    The preceding condition is equivalent to
    \begin{align*}
        \gamma \ge 2\left( 1-\beta \right)^{\frac{2}{k+1}}-1 {= T_{xy}},
    \end{align*}
    since $-1 \le \gamma \le 1$ and $0 \le \beta \le 1$. Similarly, $y \le z$ if and only if
    \begin{align*}
        1-\beta \ge \left(\frac{1+\gamma}{2}\right)^{\frac{k-1}{2}} \left(1-\frac{\beta}{2}\right),
    \end{align*}
    which is equivalent to
    \begin{align*}
        \gamma \le 2\left( \frac{1-\beta}{1-\beta/2} \right)^{\frac{2}{k-1}}-1 {= T_{zy}},
    \end{align*}
    since we assume $k>1$.
\end{proof}

Next, under different conditions, we write the total production level as a piece-wise function of $\gamma$ in the following two lemmas.
\begin{lemma}\label{lemma: A3}
    For $-1 \le \gamma \le 1$, if $T_{zy} \le T_{xy}$, then
    \begin{align*}
    \tilde{q}_S^{(k)}+\tilde{q}_B^{(k)} =
        \begin{dcases}
            z, &\text{ if } \gamma \le T_{zy},\\
            x, &\text{ if } \gamma \ge T_{xy},\\
            y, &\text{ otherwise},
        \end{dcases}
    \end{align*}
    if $T_{zy} \ge T_{xy}$, then
    \begin{align*}
    \tilde{q}_S^{(k)}+\tilde{q}_B^{(k)} =
        \begin{dcases}
            z, &\text{ if } \gamma \le T_{xy},\\
            x, &\text{ if } \gamma \ge T_{zy},\\
            x-y+z, &\text{ otherwise}.
        \end{dcases}
    \end{align*}
\end{lemma}
\begin{proof}
    Since $0 \le \beta \le 1$, both $T_{xy}$ and $T_{zy}$ are between -1 and 1. The domain of $\gamma$ is then partitioned into three intervals. Based on Lemmas \ref{lemma: A1} and \ref{lemma: A2}, assuming $T_{zy} \le T_{xy}$, we have the following observations for each interval,
    \begin{enumerate}[label=(\roman*)]
        \item if {$\gamma \le T_{zy}$:} we have $x \ge y$, $y \le z$ and thus
        $$\tilde{q}_S^{(k)}+\tilde{q}_B^{(k)} = \max\{y,z\} = z;$$
        \item if {$T_{zy} \le \gamma \le T_{xy}$:} we have $x \ge y$, $y \ge z$ and thus
        $$\tilde{q}_S^{(k)}+\tilde{q}_B^{(k)} = \max\{y,z\} = y;$$
        \item if {$\gamma \ge T_{xy}$:}  we have $x \le y$, $y \ge z$ and thus
        $$\tilde{q}_S^{(k)}+\tilde{q}_B^{(k)} = \max\{x,x-y+z\} = x.$$
    \end{enumerate}
    Assuming $T_{zy} \ge T_{xy}$, we have
    \begin{enumerate}[label=(\roman*)]
        \item if {$\gamma \le T_{xy}$:} we have $x \ge y$, $y \le z$ and thus
        $$\tilde{q}_S^{(k)}+\tilde{q}_B^{(k)} = \max\{y,z\} = z;$$
        \item if {$T_{xy} \le \gamma \le T_{zy}$:} we have $x \le y$, $y \le z$ and thus
        $$\tilde{q}_S^{(k)}+\tilde{q}_B^{(k)} = \max\{x,x-y+z\} = x-y+z;$$
        \item if {$\gamma \ge T_{zy}$:} we have $x \le y$, $y \ge z$ and thus
        $$\tilde{q}_S^{(k)}+\tilde{q}_B^{(k)} = \max\{x,x-y+z\} = x.$$
    \end{enumerate}
\end{proof}

\begin{lemma}\label{lemma: A6}
    Assume that {$T_{zy} \le T_{xy}$,} an optimal cooperation level is given by
    \begin{align*}
        \gamma^* = \max\left\{ -\frac{\beta}{2}, 2\left( 1-\beta \right)^{\frac{2}{k+1}}-1 \right\}.
    \end{align*}
\end{lemma}
\begin{proof}
    We first assume $-1 \le \gamma \le 1$. By definitions, we observe that $x$ is decreasing in $\gamma$ and $y$ is increasing in $\gamma$ since $\beta \in [0,1]$. According to Lemma \ref{lemma: A3}, the total production level $\tilde{q}_S^{(k)}+\tilde{q}_B^{(k)}$ is a piece-wise function of $\gamma$. By checking boundary cases, we verify that it is also continuous in $\gamma \in [-1,1]$. Moreover, as illustrated in Fig. \ref{fig: lemma 16}, it is flat in interval (i), increasing in interval (ii) and decreasing in interval (iii).

    \begin{figure}[ht]
    \centering
    \includegraphics[width=\linewidth]{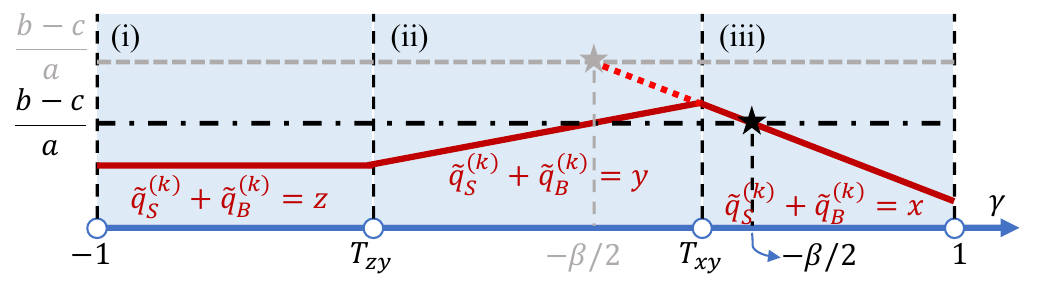}
    \caption{Illustration of the total production level, assuming $T_{zy} \le T_{xy}$.}
    \label{fig: lemma 16}
    \end{figure}
    To maximize social welfare, the total production level $\tilde{q}_S^{(k)}+\tilde{q}_B^{(k)}$ needs to be as close to $(b-c)/a$ as possible. When $\gamma = 1$, we have $x = \left( \frac{1}{2}-\frac{\beta}{4} \right)\frac{b-c}{a} \le z \le (b-c)/a$ since $\beta \in [0,1]$. Therefore, both the maximum and minimum of function $\tilde{q}_S^{(k)}+\tilde{q}_B^{(k)}$ can be achieved in interval (iii). Thus, to find an optimal $\gamma$, it suffices to study interval (iii). We let $\tilde{q}_S^{(k)}+\tilde{q}_B^{(k)} = (b-c)/a$ and solve for $\gamma$, which yields $\gamma^* = -\beta/2 < 1$. However, this solution is feasible only when the maximum of total production level is at least $(b-c)/a$, as shown in dashdotted black line in Fig. \ref{fig: lemma 16}. In the case that it is not feasible in interval (iii), $(b-c)/a$ is shown as the dashed grey line. Then we have {$\gamma^* = T_{xy} = 2\left( 1-\beta \right)^{\frac{2}{k+1}}-1$} instead. Therefore, we obtain
    \begin{align*}
        \gamma^* = \max\left\{ -\frac{\beta}{2}, 2\left( 1-\beta \right)^{\frac{2}{k+1}}-1 \right\}.
    \end{align*}

    On the other hand, when $\gamma < -1$ and $k > 1$, we have $\tilde{q}_S^{(k)}+\tilde{q}_B^{(k)}=z$. When $\gamma > 1$, $\tilde{q}_S^{(k)}+\tilde{q}_B^{(k)} \le \frac{b-c}{2a} \le z$. In both cases, the social welfare is no better than that achieved by $\gamma^*$. Therefore, $\gamma^*$ is a global optimizer.
\end{proof}

\begin{lemma}\label{lemma: A7}
    Assume that {$T_{zy} \ge T_{xy}$}, an optimal cooperation level is also given by
    \begin{align*}
        \gamma^* = \max\left\{ -\frac{\beta}{2}, 2\left( 1-\beta \right)^{\frac{2}{k+1}}-1 \right\}.
    \end{align*}
\end{lemma}

\begin{proof}
    We first assume $-1 \le \gamma \le 1$. By definitions, we have
    \begin{align*}
        x-y+z = \left[ \frac{3}{2}-\frac{\beta}{2}-\frac{(1+\gamma)^{\frac{k+1}{2}}}{2^{\frac{k+1}{2}}} \right]\frac{b-c}{a},
    \end{align*}
    which is decreasing in $\gamma$, and $x$ is also decreasing in $\gamma$ since $\beta \ge 0$. According to Lemma \ref{lemma: A3}, the total production level $\tilde{q}_S^{(k)}+\tilde{q}_B^{(k)}$ is a piece-wise function of $\gamma$. Additionally, by checking boundary cases, we see that it is continuous in $\gamma \in [-1,1]$. Thus, as shown in Fig. \ref{fig: lemma 17}, it is flat in interval (i) and increasing in intervals (ii) \& (iii).

    \begin{figure}[ht]
    \centering
    \includegraphics[width=\linewidth]{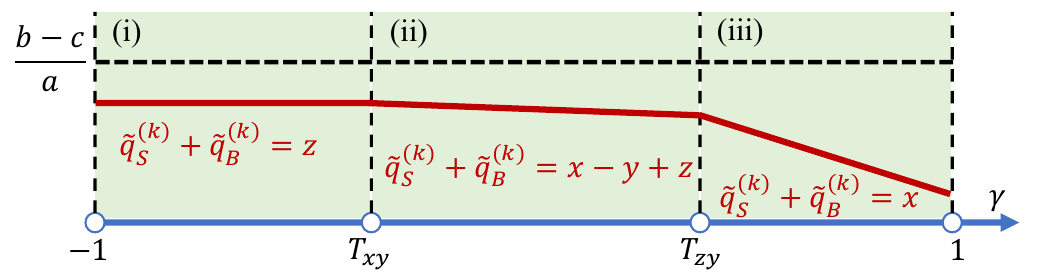}
    \caption{Illustration of the total production level, assuming $T_{zy} \ge T_{xy}$.}
    \label{fig: lemma 17}
    \end{figure}
    To maximize social welfare, the total production level $\tilde{q}_S^{(k)}+\tilde{q}_B^{(k)}$ needs to be as close to $(b-c)/a$ as possible, where $(b-c)/a$ is the dashed black line shown in Fig. \ref{fig: lemma 17}. Since $\beta \in [0,1]$, we have $z \le (b-c)/a$. Therefore, an optimal $\gamma$ is given by
    $$\gamma^* {=T_{xy}} = 2\left( 1-\beta \right)^{\frac{2}{k+1}}-1.$$
    Next, we prove that it is equivalent to the optimal cooperation level shown in the statement. Due to the condition established in the lemma, we have
    $$\left( \frac{1-\beta}{1-\beta/2} \right)^{\frac{2}{k-1}} \ge \left( 1-\beta \right)^{\frac{2}{k+1}}.$$
    Taking the power of $(k-1)/2$ on both sides yields
    $$\left( 1-\beta \right)^{\frac{k-1}{k+1}} \le \frac{1-\beta}{1-\beta/2}.$$
    Then, we relax the preceding inequality as follows,
    \begin{align*}
        \left( 1-\beta \right)^{\frac{k-1}{k+1}} \le \frac{2(1-\beta)}{1-\beta/2}.
    \end{align*}
    By eliminating $(1-\beta)$ on both sides and take their inverses, we obtain
    \begin{align*}
        -\frac{\beta}{2} \le 2\left( 1-\beta \right)^{\frac{2}{k+1}}-1.
    \end{align*}
    Hence, it is equivalent to rewrite $\gamma^*$ as
    $$\gamma^* = \max\left\{ -\frac{\beta}{2}, 2\left( 1-\beta \right)^{\frac{2}{k+1}}-1 \right\}.$$
    Moreover, if $\gamma < -1$, we have $\tilde{q}_S^{(k)}+\tilde{q}_B^{(k)}=z$. If $\gamma > 1$, $\tilde{q}_S^{(k)}+\tilde{q}_B^{(k)} \le \frac{b-c}{2a} \le z$. In both cases, the social welfare is no better than that achieved by $\gamma^*$. Therefore, $\gamma^*$ is a global optimizer.
\end{proof}

Combining Lemmas \ref{lemma: A6} and \ref{lemma: A7}, the proof of Theorem \ref{thm: optimal cooperation level}(ii) is complete.

\section{Proof of Theorem \ref{thm: utility design PoR}}\label{appendix: utility design PoR}

\begin{proof}
    Since the $\gamma^*$ shown in Theorem \ref{thm: optimal cooperation level} is optimal, the social welfare it achieves is no worse than another cooperation level $\gamma_0 \in \mathbb{R}$. Therefore, it suffices to show that the level-$k$ performance is no worse than the equilibrium performance when $\gamma = \gamma_0$. To this end, we pick $\gamma_0 = -1$. By Lemma \ref{lemma: closed-form level k behaviors utility design, -1<gamma<1}, we obtain
    \begin{align*}
        \tilde{q}_S^{(k)} + \tilde{q}_B^{(k)} = \frac{b-c}{2a} + \frac{f}{2} = q_S^{NE}+q_B^{NE},
    \end{align*}
    for any $k \in \mathbb{N}$. According to Lemma \ref{lemma: equivalent social welfare measure}, it implies $W(\tilde{q}_S^{(k)}, \tilde{q}_B^{(k)}) = W(q_S^{NE}, q_B^{NE})$. Therefore, by applying the optimal cooperation level $\gamma^*$ shown in Theorem \ref{thm: optimal cooperation level}, the achieved social welfare is at least as large as the equilibrium social welfare.
\end{proof}

\end{document}